\documentclass[12pt]{amsart}

\usepackage{amssymb,amsthm,amsmath}
\usepackage[numbers,sort&compress]{natbib}
\usepackage{color}
\usepackage{graphicx}
\usepackage{tikz}
\usepackage{amssymb,amsthm,amsmath}
\usepackage[numbers,sort&compress]{natbib}
\usepackage{color}
\usepackage{graphicx}
\usepackage{tikz}
\usepackage{xparse}

\title[ ]{Nonlinear Anderson localized \\states at arbitrary disorder}

\author{Wencai Liu}
\address[W. Liu]{ Department of Mathematics, Texas A\&M University, College Station, TX 77843-3368, USA} \email{liuwencai1226@gmail.com; wencail@tamu.edu}

\author{W.-M. Wang}
\address[W.-M. Wang]{CNRS and D\'epartement de Math\'ematique, Cergy Paris Universit\'e, 95302 Cergy-Pontoise Cedex, France } \email{wei-min.wang@math.cnrs.fr}



\theoremstyle{plain}
\newtheorem{theorem}{Theorem}[section]
\renewcommand{\l}{\left}
\renewcommand{\r}{\right}

\newtheorem{lemma}[theorem]{Lemma}

\newtheorem{remark}{Remark}

\newcommand{\Z}{\mathbb{Z}}
\newcommand{\N}{\mathbb{N}}

\newcommand{\R}{\mathbb{R}}

\theoremstyle{plain}

\newtheorem{conjecture}{Conjecture}

\begin{document}

\begin{abstract} 
It is classical, following Furstenberg's theorem on positive Lyapunov exponent for products of random SL$(2, \mathbb R)$
matrices, 
that the one dimensional random Schr\"odinger operator has Anderson localization 
at arbitrary disorder.
This paper proves a nonlinear analogue, thereby establishing a KAM-type persistence result
for a {\it non-integrable} system.
\end{abstract}

\maketitle

    \section{Introduction and the main theorem}

 We study the discrete nonlinear random Schr\"odinger equation (NLRS) in one dimension: 
 \begin{equation}\label{g4}
i\frac{\partial u}{\partial t}=-\Delta u+Vu+\delta |u|^{2p}u,\, p\in\mathbb N,
 \end{equation}
 where $\Delta$ is the discrete Laplacian:
 $$(\Delta u)(x)=u(x+1)+u(x-1),$$
 and $V=\{v_x\}$ is a family of independent identically distributed
 random variables on $[0, 1]$,  with distribution density $g$.
 Assume that $g$ is bounded, $g\in L^\infty$.
 
 Let 
  \begin{equation} \label{rs}
 H=-\Delta+V,
 \end{equation}
 be the random Schr\"odinger operator.
 It is well-known, as a consequence of Furstenberg's theorem on positive Lyapunov exponent for products of 
 random SL$(2, \mathbb R)$ matrices ~\cite{furs},
 that with probability 1, 
 $H$ has  Anderson localization, namely,  pure point spectrum with exponentially decaying eigenfunctions ~\cite{gmp,ks}.
 (See also ~\cite{bdfgvwz,jz,gok,gelzx,ms}.)
 In higher dimensions and for large disorder, i.e., replacing $V$ by $\lambda V$, $\lambda\gg1$, Anderson localization has been established using multiscale analysis ~\cite{fs}, see also ~\cite{vdk},
 or fractional moment method ~\cite{am}.
 
  Assume that $H$ has Anderson localization. Let $\{{\phi}^V_j\}_{j\in\Z}$ be the  (real) eigen-basis of $H$.
  Assume that  ${\ell}_j^V$ 
 satisfies 
 \begin{equation*}
 |{\phi}_j^V({\ell}_j^V)|=\max_{x} |{\phi}_j^V(x)|.
 \end{equation*}
(If it is not unique, one may choose a maximum arbitrarily.)
 We call  ${\ell}_j^V$, the {\it localization center}. It suffices to say here that as a consequence of 
 localization, there is an  eigenfunction labelling such that 
 when $j_1<j_2$,  $\ell_{j_1}^V\leq \ell_{j_2}^V$ (see sect.~3 and appendix \ref{label1} for details), and that we use this labelling. 


So for a given $V$ such that $H$ has Anderson localization,  let $j\in\mathbb Z$, and denote by $\phi_j^V$ and $\mu_j^V$ the eigenfunctions and corresponding eigenvalues  of $H$. 
 When $\delta=0$, all solutions to \eqref{g4} are of the form 
 $$\sum_{j\in\mathbb Z} c_j e^{-i\mu_j^Vt} \phi_j^V,$$
 with appropriate $c_j$, which decay to $0$, as $j\to\pm\infty$. 
This paper is concerned with the case $c_j\neq 0$ for finite (but arbitrary) number of $j$. 

\smallskip

Denote by $\mathbb P (\cdot)$ the measure of a set. 
We prove the following nonlinear analogue:
  
\begin{theorem}\label{mainthm}
Consider the discrete NLRS in one dimension:
 \begin{equation}\label{g28} 
i\frac{\partial u}{\partial t}=-\Delta u+Vu+\delta |u|^{2p}u, \, p\in\mathbb N.
\end{equation}
Let $a=(a_1,a_2,\cdots,a_b)\in [1,2]^b$.
For any $\epsilon>0$, there exists $l_{\epsilon}$ such that the following holds. Fix  any $L\geq \ell_\epsilon$ and $\beta_k\in\Z$, $k=1,2,\cdots,b$  satisfying 
$10L\leq |{\beta}_k|\leq L^3$ and $|{\beta}_{k}-{\beta}_{k'}|\geq  10L$ for any distinct $k,k'\in\{1,2,\cdots,b\}$,  there exist a probability space $X_{\epsilon}$ with  
$\mathbb{P}(X_\epsilon)\geq 1-\epsilon$ and   $\delta_0>0$ (depending on $g$, $\epsilon$ and  $L$) such that  for any $V\in X_{\epsilon}$ and $0<\delta\leq \delta_0$,  there exists 
 a set $A_{\delta}\subset [1, 2]^b$ of measure at least $1-e^{-|\log\delta|^{1/2}}$, such that for any $a\in A_\delta$, any eigenfunction $\phi_{\alpha_k}^V$ with $\ell_{\alpha_k}^V\in B_k=\{l\in\Z: |l-{\beta}_k|\leq L\}$,  $k=1,2,\cdots,b$, 
the solution to the linear equation:
\begin{equation*}
u_0(t, x)=\sum_{k=1}^b a_k e^{-i \mu_{\alpha_k}^V t}\phi_{\alpha_k}^{V}(x),
\end{equation*}
bifurcates to a solution to the nonlinear equation \eqref{g28}:
\begin{equation*}
u(t, x)=\sum_{(n,j)\in \Z^b\times \Z} \hat u(n,j)e^{i n\cdot\omega t}\phi_{j}^{V}(x)=\sum_{k=1}^b a_k e^{-i \omega_kt}\phi_{\alpha_k}^{V}(x)+O(\delta^{1/2}),
\end{equation*}
satisfying
$\omega=(\omega_1, \omega_2, \cdots, \omega_b)=(\mu_{\alpha_1}^V,\mu_{\alpha_2}^V,\cdots, \mu_{\alpha_b}^V)+O(\delta)$,
and $\hat u(n, j)$ decay exponentially as 
$|(n, j)|\to\infty$. 
\end{theorem}
\begin{remark}
	 \begin{enumerate}
		{\rm \item 	For any $B_k$, $k=1,2,\cdots,b$, there are at least $2(1-\epsilon) L$ normalized eigenfunctions $\phi_{\alpha_k}^V$ such that the localization centers 
		$\ell_{\alpha_k}^V$ lie  in $B_k$.
		\item We could replace $L^3$ with $e^{L^{\kappa}}$, $0<\kappa<1$.  We could also replace $1-e^{-|\log\delta|^{1/2}}$ with
		$1-e^{-|\log\delta|^{\kappa}}$, $0<\kappa<1$. }
	\end{enumerate}

	\end{remark}

\subsection{About Theorem \ref{mainthm}} The linear solution $u_0$ is localized in space, and quasi-periodic in time, with frequencies 
the $b$ eigenvalues of the linear random Schr\"odinger operator. Theorem \ref{mainthm} shows that under small nonlinear 
perturbations, for a large set of amplitudes, there is a solution to the nonlinear equation nearby. This nonlinear solution $u$
remains localized in space and quasi-periodic in time; moreover the frequencies are harmonics of the modulated $b$ Fourier modes of the linear random Schr\"odinger equation.

Theorem \ref{mainthm} is a KAM-type persistence result. Most such results pertain to perturbations of integrable systems.
The random Schr\"odinger equation is, however, {\it non-integrable}. Nonetheless, Theorem \ref{mainthm} shows persistence
of time quasi-periodic, localized solutions. Moreover there is an abundance of such solutions. This is the main novelty.
\smallskip

\begin{remark} {\rm Previously the paper ~\cite{bw} established existence of quasi-periodic solutions at large disorder. It perturbs about the diagonal operator $\lambda V$, $\lambda\gg1$, in the canonical $\mathbb Z^d$ basis. The method is not applicable here.
For $d=1$ and large disorder, see also ~\cite{zg}. For other nonlinear random models, see e.g., ~\cite{fsw}.}
\end{remark}

\begin{remark} {\rm The symbol $\hat u$ is merely a notation here and does not carry the connotation of being the dual of $u$.}
\end{remark}

\subsection{Ideas of the proof} One of the main ideas is to {\it fix} a ``good" random potential, and work in the eigenfunction basis provided by the
random Schr\"odinger operator. 
We give the requirements to be good in sect.~ \ref{good} below. Fixing a potential circumvents the lack of control of the eigenfunctions as the potentials vary. 
Note moreover that generally it is not possible to know precisely the eigenfunctions of the random Schr\"odinger operators, even for large disorder, see ~\cite{ek}. 

So fix indeed such a good potential $V$, and let
\begin{equation}\label{u0}
u_0(t, x)=\sum_{k=1}^b a_k e^{-i \mu_{\alpha_k}^V t}\phi_{\alpha_k}^{V}(x),
\end{equation}
be a solution to the linear equation, as in Theorem \ref{mainthm}. As an ansatz, we seek solutions of the form: 
\begin{equation}\label{ansatz}
u(t, x)=\sum_{(n,j)\in \Z^b\times \Z} \hat u(n,j)e^{i n\cdot\omega t}\phi_{j}^{V}(x).
\end{equation}
Note that $u(t, x)$ of the above form are closed under multiplication and complex conjugation. So we may seek solutions to \eqref{g4} in this form.

Using \eqref{ansatz} in \eqref{g4} leads to the following nonlinear system of equations on $\mathbb Z^{b}\times \mathbb Z$:
\begin{equation}\label{meq}
(n\cdot\omega+\mu_j^V)\hat u(n, j)+\delta W_{\hat u}(n, j)=0, \, (n, j)\in\mathbb Z^{b}\times\mathbb Z,
\end{equation}
where, to give an idea, when $p=1$, 
\begin{align}
	W_{\hat u}(n, j)=&\sum_{n_1+n_2-n_3=n\atop{n_1,n_2,n_3\in\Z^b}} \sum_{j_1,j_2,j_3\in\Z } \hat{u}(n_1,j_1) \hat{u}(n_2,j_2) {\overline{\hat{u}(n_3,j_3)}}\nonumber\\
	&\left(\sum_{x\in\Z}\phi_{j}^V(x)\phi_{j_1} ^V(x) {\phi_{j_2}^V(x)} \phi_{j_3}^V(x) \right);\label{W3}
\end{align}
while for general $p$, 
\begin{align}
W_{\hat u} (n, j)=&\sum_{n'+\sum_{m=1}^p (n_m-n_m')=n\atop{n', n_m,n_m'\in\Z^b}} \sum_{j', l_m,l_m'\in\Z } \hat{u}(n',j')\prod_{m=1}^p\hat {u}(n_m,l_m) \overline{\hat{u}(n'_m,l'_m) }\nonumber\\
&\left(\sum_{x\in\Z}\phi_{j}^V(x)\phi_{j'}^V(x)\prod_{m=1}^p\phi_{l_m} ^V(x) {\phi_{l_m'}^V(x)}\right).\label{W}
\end{align}

\subsection{The good potentials}\label{good} The linear solution $u_0$ solves \eqref{g4} to order $\delta$. One may write $u_0$ in the form \eqref{ansatz},
with $\hat u(-e_k, \alpha_k)=a_k$, $k=1, 2, \cdots, b$, where $e_k$ is the $k$th basis vector of $\mathbb Z^b$, $\hat u(n, j)=0$ otherwise, and $\omega_k=\mu_{\alpha_k}^V$, $k=1, 2, \cdots, b$.  The 
vector $W_{\hat u}$ in \eqref{meq} depends on $a_k$, $k=1, 2, \cdots, b$. Generally speaking, one would need parameters to solve the nonlinear equation \eqref{meq} using a Newton scheme, 
starting from the approximate solution $u_0$. Since $V$ is fixed, the $a_k$'s are the parameters in the problem. There is however, a $\delta$ factor in front of $W_{\hat u}$. 

The small $\mathcal O(\delta)$ parameters pose difficulties mainly at small scales: $|(n, j)|\ll \delta^{-1}$,
when estimating the inverse of the linearized operators. The key new idea is that we can overcome this difficulty if the diagonals of the linear operator in \eqref{meq} satisfy a {\it clustering property}. Roughly speaking, this means that
if two diagonal elements are ``not close", then they are ``far apart".  (One may think of the integers, which have this property: if two integers 
are not equal, then they are at least of distance $1$.) This then permits localizing about the diagonals in $\mathcal O(\delta)$ intervals, which compensates
for the small $\mathcal O(\delta)$ parameters. The potentials $V$ that lead to clustering properties, in addition to Anderson localization, are {\it good} potentials.

It should be emphasized that the clustering property is only needed at {\it small} scales, and {\it not} large ones.
This makes the approach robust, potentially applicable to many problems.

\begin{remark}{\rm See papers ~\cite{wduke} and ~\cite{wkg}, which use deterministic clustering properties. In ~\cite{wkg}, this was established using number theory.}
\end{remark}

\subsection{Anderson localization and clustering property of the diagonals} We use Anderson localization to establish probabilistic clustering at small scales. 
So one may set $\omega$ to be the frequencies of the $u_0$ in \eqref{u0}, which are $b$ eigenvalues of the $H$ in \eqref{rs}. The diagonals in \eqref{meq} then correspond to a family
of harmonics, i.e., certain linear combinations of the eigenvalues of $H$.

The proof of the clustering of these (low lying) harmonics is rather delicate. The Minami estimate ~\cite{mi}
on eigenvalue spacing plays a fundamental role. Uniform property of  Anderson localization, see ~\cite{djls,gk},
is essential. 
Wegner estimate ~\cite{we} comes into play as well. This is done in sects.~2 and 3, and the conclusion is summarized in Theorem \ref{mainkeythm}, which also provides lower bounds on the diagonals.
The clustering property permits the analysis to go beyond small perturbative scales, and is one of the main points of the paper.

\subsection{Small scale analysis} The clustering property indicates that at {\it small scales} the spectrum of the diagonal operator has many {\it gaps}. 
Using perturbation theory, the linearized operators are invertible in the gaps; while away from the gaps, one may work locally in intervals of size $\mathcal O(\delta)$. 
This greatly reduces the number of resonances, and consequently $\mathcal O(\delta)$ parameters suffice for the analysis. 
This is the case for the proof of the large deviation theorem applying  Cartan estimates in sect.~\ref{ldt}, as well as for the semi-algebraic projection in sect. \ref{proj}.

\subsection{Large scale analysis}  Large scale analysis is related to what has been done before in ~\cite{bw}, cf. also Chap.~18 ~\cite{bbook},
which are a priori tailored for parameters of order $1$. However, after incorporating the local argument recounted above, it 
can be adapted and used to prove Theorem \ref{mainthm}.

\subsection{Organization of the paper} Sect. 2 establishes the {\it good} potential space;  sect.~3 makes linear estimates for small scales; sect.~4 proves a large deviation theorem,
to be used for the nonlinear analysis at large scales; sects.~5 and 6 finally solve \eqref{meq} and hence \eqref{g4}, using a Lyapunov-Schmidt decomposition and Theorems \ref{mainkeythm} and \ref{thmldt}. 

 \section{One dimensional random Schr\"odinger operators\\ in finite volumes}
For an operator $H$ on $\ell^2(\Z^d)$ and  $\Lambda\subset \Z^d$, let $H_{\Lambda}=R_{\Lambda} HR_{\Lambda}$, where $ R_{\Lambda} $ is the restriction.
For $n=(n_1,n_2,\cdots,n_d)\in \Z^d$, let $|n|=\max_{j\in\{1,2,\cdots,d\}}|n_j|$ denote the $\ell^\infty$ norm. In this paper $d$  equals either 1 or $b+1$.

In this section, we study the  one dimensional random Schr\"odinger operator  $H=-\Delta+V$ restricted to finite volumes.
For $\Lambda \subset \Z$,
 denote by 
 $\tilde{\mu}^{\Lambda}_j$, $j\in\Lambda$, eigenvalues of $H_{\Lambda}=R_{\Lambda} HR_{\Lambda}$, with corresponding normalized eigenvectors $\tilde{\phi}_j^{\Lambda}$.
 \begin{remark}
 	{\rm Note that $\tilde{\mu}^{\Lambda}_j$  and   $\tilde{\phi}_j^{\Lambda}$ depend on the realization of the potentials in $\Lambda$. It is convenient here  to label the eigenvalues 
	and normalized eigenvectors by $j\in\Lambda$, instead of $j\in \{1, 2, \cdots, |\Lambda|\}$.}
	 \end{remark}


For a  ball  $B=\{n\in\Z:|n-n_0|\leq l\}$ of size $l$ with center $n_0$, denote by $rB$, the dilation: $rB=\{n\in\Z:|n-n_0|\leq rl\}$.
Fix $L>0$.
Let $\tilde{B}_{(k,L)}$ be a ball of size $L$, $k=1,2,\cdots,b$. Assume that 
 \begin{equation*}
L\leq  {\rm dist} (\tilde{B}_{(k,L)},0)\leq  L^4, 
 \end{equation*}
 and 
 for any distinct $k$ and $k'$,
\begin{equation}\label{g30}
{\rm dist} (\tilde{B}_{(k,L)},\tilde{B}_{(k',L)})\geq  6L.
\end{equation}
Since $L$ will be fixed, we omit the dependence on $L$ of $\tilde{B}_{(k,L)}$ and write simply  $\tilde{B}_{k}$.

Recall that  $e_k$, $k=1, 2, \cdots, b$, are the canonical basis vectors of $\mathbb Z^b$.
Let $C_0>0,\gamma_0>0,q_0>0$ 
be three fixed constants, which will be determined later. 
Assume that $H_{\Lambda}$ has  $b$ eigen-pairs  $\tilde{\mu}^{\Lambda}_{\tilde{\alpha}_k}$ and $\tilde{\phi}^{\Lambda}_{\tilde{\alpha}_k}$, $k=1,2,\cdots,b$ such that   for any $k=1,2,\cdots,b$, there exists 
  $\tilde{\ell}_{\tilde{\alpha}_k}^\Lambda\in\tilde{B}_k$  such that
  \begin{equation}\label{g18}
|\tilde{\phi}^\Lambda_{\tilde{\alpha}_k}(\ell)|\leq C_0(1+|\tilde{\ell}_{\tilde{\alpha}_k}^\Lambda|)^{q_0}e^{-\gamma_0|\ell-\tilde{\ell}_{\tilde{\alpha}_k}^\Lambda|},  \ell\in \Lambda.  
\end{equation} 
Let $\tilde{\omega}^\Lambda_k=\tilde{\mu}_{\tilde{\alpha}_k}^\Lambda$, $k=1,2,\cdots,b$ and  $\tilde{\omega}^\Lambda=(\tilde{\omega}_1^\Lambda,\tilde{\omega}_2^\Lambda,\cdots,\tilde{\omega}_b^\Lambda)\in\R^b$.
When there is no ambiguity, we omit the dependence on $\Lambda$. In the following, $\delta>0$ is sufficiently small.

\begin{remark} {\rm Condition \eqref{g18} is motivated by uniform properties of Anderson localization in infinite volume, see Theorem \ref{thm1}.}
\end{remark}

\subsection{Estimates on the diagonals}
Let $\Lambda_1=\left[-2\left \lfloor e^{|\log\delta|^\frac{3}{4}} \right \rfloor, 2\left \lfloor e^{|\log\delta|^\frac{3}{4}} \right \rfloor\right]$, where $ \lfloor x \rfloor$  is the integer part of $x$. 
(The choice of scales is in view of the later nonlinear analysis.)
Denote by $S_1$ the probability space on which $H_{\Lambda_1}$ has  eigen-pairs  $\tilde{\mu}^{\Lambda_1}_{\tilde{\alpha}_k}$ and $\tilde{\phi}^{\Lambda_1}_{\tilde{\alpha}_k}$, $k=1,2,\cdots,b$  satisfying \eqref{g18} (with $\Lambda=\Lambda_1$), and there exists either
\begin{equation}\label{set}
(n, j)\in  \left[-2\left \lfloor e^{|\log\delta|^\frac{3}{4}} \right \rfloor, 2\left \lfloor e^{|\log\delta|^\frac{3}{4}} \right \rfloor\right]^{b+1}\backslash\{- e_k, \tilde{\alpha}_k\}_{k=1}^b ,
\end{equation}
such that  
\begin{equation}\label{g5}
|n\cdot\tilde{ \omega}^{\Lambda_1}+\tilde{\mu}_j^{\Lambda_1}|\leq 4\delta^{\frac{1}{8}} 
\end{equation}
or 
\begin{equation}\label{set1}
	(n, j)\in  \left[-2\left \lfloor e^{|\log\delta|^\frac{3}{4}} \right \rfloor, 2\left \lfloor e^{|\log\delta|^\frac{3}{4}} \right \rfloor\right]^{b+1}\backslash\{ e_k, \tilde{\alpha}_k\}_{k=1}^b ,
\end{equation}
such that  
\begin{equation}\label{g511}
 	|-n\cdot \tilde{\omega}^{\Lambda_1}+\tilde{\mu}_j^{\Lambda_1}|\leq 4\delta^{\frac{1}{8}}
\end{equation} 
and the normalized  eigenvector corresponding to $\tilde{\mu}_j^{\Lambda_1}$ satisfies  for some $\tilde{\ell}_{j}^{\Lambda_1}\in \Lambda_1$, 
	\begin{equation}\label{g22}
	|\tilde{\phi}^{\Lambda_1}_{j}(\ell)|\leq C_0 (1+|\tilde{\ell}^{\Lambda_1}_{j}|)^{q_0}e^{-\gamma_0|\ell-\tilde{\ell}_{j}^{\Lambda_1}|},\ell\in \Lambda_1.
	\end{equation}

Fix a large constant $q_1>0$.
Denote by $S_N$ the probability space such that for any  $j\in[-N,N], j'\in  [-N,N]$ and $j\neq j'$,
\begin{equation}\label{g31}
	|\tilde{\mu}_j^\Lambda-\tilde{\mu}_{j'}^\Lambda|\geq \frac{1}{N^{q_1}}, \Lambda=[-N,N].
\end{equation}
 Denote by $\mathbb{P}$ the measure as before, and $\mathbb{E}$ the expectation.
\begin{theorem}\label{thm2}
For small $\delta$, we have
	\begin{equation*}
		\mathbb{P}(	  S_1 )\leq e^{C|\log\delta|^\frac{3}{4}} \delta^{\frac{1}{8}},
	\end{equation*} 	
where $C$ is a large constant independent of $\delta$.
\end{theorem}

We will use the following three lemmas to prove  Theorem  \ref{thm2}.
\begin{lemma}[Wegner estimate ~\cite{we}]\label{lem2}
	Let $\Lambda\subset \Z$.
	For any $E\in\R$ and $\varepsilon>0$,
	\begin{equation*}
		\mathbb{E}({\rm dist} (E,\sigma(H_{\Lambda}))\leq \varepsilon) \leq C |\Lambda|\varepsilon.
	\end{equation*}
\end{lemma}

\begin{lemma}[Minami estimate ~\cite{mi}]\label{lem3}
	Let $\Lambda\subset \Z$ and $J\subset\R$ be an interval. Then we have 
	\begin{equation*}
		\mathbb{E}([{\rm tr}({\bf 1}_J(H_{\Lambda}))]\cdot[{\rm tr}({\bf 1}_J(H_{\Lambda}))-1])\leq C |\Lambda| ^2|J|^2,
	\end{equation*}
	where ${\bf 1}_J$ is the characteristic function of the  interval $J$. 
	In particular, we have that
	\begin{equation*}
		\mathbb{P}(\text{there exist distinct } j,j^\prime \in\Lambda \text{ such that } |\tilde{\mu}_j^{\Lambda} -\tilde{\mu}_{j'}^{\Lambda}|\leq \varepsilon)\leq C\varepsilon |\Lambda| ^2,
	\end{equation*}
and hence
	\begin{equation*}
	\mathbb{P}(S_N)\leq  CN^{-q_1+2}.
\end{equation*}
\end{lemma}

 \begin{lemma}\label{lemeig}
 Let $\Lambda\subset \Z$. Assume that  eigenvalues $\tilde{\mu}_j$, $j\in\Lambda$, of $H_{\Lambda}$  are simple and  let
 \begin{equation}\label{g41}
d=\min_{j\neq j',j\in\Lambda,j'\in\Lambda} |\tilde{\mu}_j-\tilde{\mu}_{j'}|>0.
 \end{equation} 
Let $ (\tilde{\mu},\tilde{\phi})$ be an eigen-pair of $H_{\Lambda}$. 
Let $l\in \Lambda$, and 
$ (\tilde{\mu}^s,\tilde{\phi}^s)$ (depending continuously on $s$)  be an  eigenpair of $H_{\Lambda}+s I_{\{l\}}$ with $|s|\leq \frac{d}{10}$ satisfying  $ \tilde{\mu}^s|_{s=0}=\tilde{\mu}$  and $\tilde{\phi}^s|_{s=0}=\tilde{\phi}$. Then
	\begin{equation*}
\frac{d	\tilde{\mu}^s}{ds}=|\tilde{\phi}(l)|^2+|s|^2O\left(\frac{|\Lambda|}{d^2}\right),
\end{equation*}
and hence
	\begin{equation*}
\tilde{\mu}^s-\tilde{\mu}=s|\tilde{\phi}(l)|^2+|s|^3O\left(\frac{|\Lambda|}{d^2}\right).
\end{equation*}
 \end{lemma}
\begin{proof}

				Let $\{(\tilde{\mu}_j,\tilde{\phi}_j),\, j\in\Lambda\}$ be the complete set of eigen-pairs of  $H_{\Lambda}$.
		Without loss of generality, assume that $\tilde{\mu}=\tilde{\mu}_1$ and $\tilde{\phi}=\tilde{\phi}_1$.

			By  a standard perturbation argument and \eqref{g41},  one has that  for any $|s|\leq \frac{d}{10 }$,  
			$$	|\tilde{\mu}^s-\tilde{\mu}_{1}|\leq \frac{d}{2}$$
			 and  
			 \begin{equation}\label{g42}
			|\tilde{\mu}^s-\tilde{\mu}_{j'}|\geq \frac{d}{2}, j'\in \Lambda\backslash\{1\} .
			\end{equation} 
	Let
	\begin{equation*}
\tilde{\phi}^s =\sum_{j\in \Lambda} c_j^s \tilde{\phi}_j.
	\end{equation*}
	Then
	\begin{equation}\label{g56}
	(H_{\Lambda}+s I_{\{l\}})\tilde{\phi}^s=\tilde{\mu}^s \tilde{\phi}^s= \sum_{j\in\Lambda}\tilde{\mu}^s  c_j^s \tilde{\phi}_j
	\end{equation}
	and
		\begin{equation}\label{g57}
	H_{\Lambda}\tilde{\phi}^s=\sum_{j\in\Lambda} c_j^s  \tilde{\mu}_j\tilde{\phi}_j.
	\end{equation}
	Since $ ||	(H_{\Lambda}+s I_{\{l\}})\tilde{\phi}^s -	H_{\Lambda}\tilde{\phi}^s||=O(s)$, by \eqref{g42}, \eqref{g56} and \eqref{g57},  one has that  for any $j\in \Lambda\backslash \{1\}$,
	\begin{equation*}
	|c_j^s| \leq O\l(\frac{s}{d}\r).
	\end{equation*}
This implies that  $ 1-O\l(\frac{|\Lambda| s^2}{d^2}\r)\leq (c_1^s)^2\leq 1$.
	Therefore, one has that
	\begin{equation}\label{gg11}
|\tilde{\phi}^s(l)|^2=|\tilde{\phi}(l)|^2+O\l(\frac{|\Lambda|s^2}{d^2}\r).
	\end{equation}
	By eigenvalue variations and \eqref{gg11}, one has that 
	\begin{align*}
	\frac{d	\tilde{\mu}^s}{ds}&=\langle \tilde{\phi}^s,	\frac{d(H_{\Lambda}+s I_{\{l\}})}{ds}\tilde{\phi}^s \rangle\\
&=|\tilde{\phi}^s(l)|^2\\
&=|\tilde{\phi}(l)|^2+O\l(\frac{|\Lambda|s^2}{d^2}\r).
	\end{align*}
We conclude that
	\begin{equation*}
	\tilde{\mu}^s-\tilde{\mu}=s|\tilde{\phi}(l)|^2+|s|^3O\left(\frac{|\Lambda|}{d^2}\right).
	\end{equation*}

\end{proof}


For any $n=(n_1,n_2,\cdots,n_b)\in \Z^b$, denote by $$\text { supp }  n=\#\{n_k:n_k\neq 0,k=1,2,\cdots,b\},$$
where $\#$ denotes the number of elements in a set.
\begin{proof}[\bf Proof of Theorem \ref{thm2}]
	
	Let $N=2\left \lfloor e^{|\log\delta|^\frac{3}{4}} \right \rfloor$
	and  $\tilde{S}_1=S_{N}$. By Lemma \ref{lem3} (Minami estimate), it suffices to prove that 
		\begin{equation*}
		\mathbb{P}(	  S_1 \cap \tilde{S}_1)\leq e^{C|\log\delta|^\frac{3}{4}} \delta^{\frac{1}{8}}.
	\end{equation*} 	
	Since the total number of $(n, j)$  in \eqref{set} and \eqref{set1}  is bounded by  $(2N+1)^{b+1}$, it suffices to prove that for a fixed 
	$(n, j)$ in the set, 
	with probability  at most 
	$ e^{C|\log\delta|^\frac{3}{4}} \delta^{\frac{1}{8}}$,  either  \eqref{g5} or \eqref{g511} holds.
Without loss of generality, we only  consider probability space $P_{n,j}$ ($P$ for simplicity) such that \eqref{g22} holds, and
	\begin{equation}\label{g21}
 |n\cdot\tilde{ \omega} +\tilde{\mu}_j|\leq 4\delta^{\frac{1}{8}} .
	\end{equation}
Our goal is to show that 
\begin{equation}\label{g71}
\mathbb{P}(P_{n,j} \cap \tilde{S}_1 )\leq e^{C|\log\delta|^\frac{3}{4}} \delta^{\frac{1}{8}}. 
\end{equation}
	
	{\bf Case 1: $n=0$.}  	In this case,  Wegner  estimate  implies \eqref{g71}.

	{\bf Case 2: } $\text{supp }n=1$
	
	{\bf Case $2_1$: } $\text{supp }n=1$
	and $n=-e_k,k=1,2,\cdots,b$. 

	Without loss of generality,   assume that $n=(-1, 0, \cdots, 0)$. In this case,
	$n\cdot\tilde{ \omega}+\tilde{\mu}_j=-\tilde{\mu}_{\tilde{\alpha}_1}+\tilde{\mu}_j$ and $j\neq \tilde{\alpha}_1$.  This is  impossible by \eqref{g31}.

	{\bf Case $2_2$: } $\text{supp }n=1$, $n=e_k,k=1,2,\cdots,b$ or $n=re_k,k=1,2,\cdots,b$  with $|r|\geq 2$.
		
	Without loss of generality, assume that $n=(n_1,0,\cdots,0)$ with $n_1=1$ or $|n_1|\geq 2$. 
	
Denote by $\tilde{\phi}_{\tilde{\alpha}_1} $ and $\tilde{\phi}_j$  eigenvectors of eigenvalues $\tilde{\mu}_{\tilde{\alpha}_1}=\tilde{\omega}_1$ and $\tilde{\mu}_j$.
For any $l\in 2\tilde{B}_1$, denote by $P_{l}$ the probability space such that 
\begin{equation}\label{g33}
|\tilde{\phi}_{\tilde{\alpha}_1} (l)|^2\geq \frac{4}{5} |\tilde{\phi}_{j} (l)|^2+\frac{1}{L^{10}}.
\end{equation}	
By \eqref{g18}, one has that
\begin{equation*}
\sum_{\ell \notin 2\tilde{B}_1}	|\tilde{\phi}_{\tilde{\alpha}_1} (\ell)|^2\leq e^{-\frac{\gamma_0}{2}L},
\end{equation*}
and hence
\begin{equation}\label{g72}
\sum_{\ell \in 2\tilde{B}_1} |\tilde{\phi}_{\tilde{\alpha}_1} (\ell)|^2\geq 1-e^{-\frac{\gamma_0}{2}L}.
\end{equation}
By \eqref{g33} and \eqref{g72}, one has
that
\begin{equation}\label{g76}
P_{n,j}\subset \bigcup_{l\in2 \tilde{B}_1}P_l.
\end{equation}

 Take any $l\in 2 \tilde{B}_1 $ such that \eqref{g33} holds.
Split $[0,1]$ into $N^{4q_1}$ intervals of size $N^{-4q_1}$ and  take any interval $I=[f_1,f_2]$. 
Define the probability space $P_{l, I}$ to be $P_{l, I}=\{V\in P_l: V_l\in I\}$.
 Applying  Lemma \ref{lemeig}, one has that if we fix $V_{\ell}$,  $ \ell\in\Lambda_1\backslash\{l\}$, then for any $f\in I$, 
	\begin{equation}\label{g73}
\frac{d\tilde{\mu}_{\tilde{\alpha}_1}^{V_l=f}}{df}=|\tilde{\phi}^{V_l=f_1}_{\tilde{\alpha}_1}(l)|^2+  N^{-5q_1}  O(1),
\end{equation}
and 
		\begin{equation}\label{g74}
		\frac{d\tilde{\mu}_{j}^{V_l=f}}{df}=|\tilde{\phi}^{V_l=f_1}_{j}(l)|^2+  N^{-5q_1}  O(1),
	\end{equation}
where $\tilde{\mu}_{i}^{V_l=f}$ is the eigenvalue with the potential  $V_l=f$ and fixed $V_{\ell}$,  $ \ell\in\Lambda_1\backslash\{l\}$.

Obviously, 
\begin{equation}\label{g93}
	\frac{d\tilde{\mu}_{j}^{V_l=f}}{df}\geq 0.
\end{equation}

By \eqref{g33},  \eqref{g73},  \eqref{g74} and \eqref{g93}, one has that
\begin{align}
	\left |	\frac{d (n\cdot\tilde{ \omega}^{V_l=f}+\tilde{\mu}_j^{V_l=f})}{df}\right |&\geq \frac{3}{4}|\tilde{\phi}^{V_l=f_1}_{\tilde{\alpha}_1}(l)|^2-  N^{-5q_1}  O(1)\nonumber\\
	&	\geq \frac{3}{4L^{10}}- N^{-5q_1}  O(1)\nonumber\\
		&	\geq \frac{1}{2L^{10}}.\label{g75}
\end{align}

Since $g\in L^\infty [0, 1]$, by \eqref{g21} and \eqref{g75}, one has that for any $ {l\in2 \tilde{B}_1}$, 
 \begin{equation}\label{g77}
\mathbb{P}( 	P_{n,j} \cap P_l\cap \tilde{S}_1) \leq O(1)N^{4q_1} L^{10}\delta^{\frac{1}{8}}.
 \end{equation}
	By \eqref{g76} and \eqref{g77}, one has that
	 \begin{equation}\label{g78}
		\mathbb{P}( 	P_{n,j} \cap \tilde{S}_1) \leq O(1)N^{4q_1} L^{11}\delta^{\frac{1}{8}}.
	\end{equation}
It implies \eqref{g71} since $\delta$ is sufficiently small (depending on $L$).

	{\bf Case 3:} $\text{supp }n\geq 2$
	
Let  $n_{i}$ be such that $|n_i|=\max_{i\in\{1,2,\cdots,b\}}|n_i|$. 

	{\bf Case $3_1$:} $|n_i|\geq 2$
	
Without loss of generality, assume that $n_i=n_1$.
It is easy to see 
that for any $\ell\in 2\tilde{B}_1$
\begin{equation}\label{g81}
|\tilde{\phi}_{\tilde{\alpha}_k}(\ell)|^2\leq e^{-\gamma_0 L},k=2,3,\cdots, b.
\end{equation}
	For any $l\in 2\tilde{B}_1$, denote by $P_l^1$ the probability space such that 
	\begin{equation}\label{g34}
	|\tilde{\phi}_{\tilde{\alpha}_1} (\ell)|^2\geq \frac{4}{5} |\tilde{\phi}_{j} (\ell)|^2+L^2  \l(\sum_{k=2}^b\tilde{\phi}_{\tilde{\alpha}_k} (\ell)|^2 \r)+\frac{1}{L^{10}}.
	\end{equation}	
By \eqref{g72}, \eqref{g81} and \eqref{g34}, one has that
\begin{equation}\label{g7}
	P_{n,j}\subset \bigcup_{l\in2 \tilde{B}_1}P_l^1.
\end{equation}
	Replacing \eqref{g76} with \eqref{g7}, and 
 following the proof of  Case $2_2$ (using Lemma \ref{lemeig}), we still have \eqref{g71} .
	
 	{\bf Case $3_2$:} $|n_i|\leq1$, namely $n_k=0,\pm1$, $k=1,2,\cdots,b$.
 	
 	In this case, clearly,  there exist at least two non-zero $n_k$, $k=1,2,\cdots, b$.
 	
 	Without loss of generality, assume that $n_1\neq0$ and $n_2\neq 0$.
 	In this case, if $\tilde{\ell}_j\geq 2L^4$, by \eqref{g22}, one has that for any $\ell\in\cup_{k=1}^b 2\tilde{B}_k $,
		\begin{equation}\label{g29}
|\tilde{\phi}_{j}(\ell)|\leq e^{- \gamma_0 L}.
\end{equation}
If $\tilde{\ell}_j\leq 2L^4$,   by \eqref{g30} and \eqref{g22}, one has that either for any $\ell\in 2\tilde{B}_1$, 
\begin{equation}\label{g39}
|\tilde{\phi}_{j}(\ell)|\leq e^{- \gamma_0 L},
\end{equation}
or 
 for any $\ell\in 2 \tilde{B}_2$, 
\begin{equation}\label{g40}
|\tilde{\phi}_{j}(\ell)|\leq e^{- \gamma_0 L}.
\end{equation}
Without loss of generality assume that \eqref{g39} holds. Therefore,
 for any $\ell \in 2\tilde{B}_1$,
\begin{equation}\label{g82}
|\tilde{\phi}_i(\ell)|^2\leq e^{-\gamma _0L},i=\tilde{\alpha}_2,\tilde{\alpha}_3,\cdots, \tilde{\alpha}_b, \text{ and } i=j.
\end{equation}
For any $l\in 2\tilde{B}_1$, denote by $P_l^2$ the probability space such that 
\begin{equation}\label{g43}
|\tilde{\phi}_{\tilde{\alpha}_1} (l)|^2\geq L^{2}  |\tilde{\phi}_{j} (l)|^2+L^{2}\l(\sum_{k=2}^b\tilde{\phi}_{\alpha_k} (l)|^2 \r)+\frac{1}{L^{10}}.
\end{equation}	
By \eqref{g72}, \eqref{g82} and \eqref{g43}, one has that
\begin{equation}\label{g80}
	P_{n,j}\subset \bigcup_{l\in2 \tilde{B}_1}P_l^2.
\end{equation}
Replacing \eqref{g76} with \eqref{g80}, and 
following the proof of  Case $2_2$ or Case $3_1$ (also using Lemma \ref{lemeig}), we   have \eqref{g71} .
\end{proof}

\subsection{Spacing of the diagonals}
For the purpose of nonlinear analysis, it suffices to work with scales $|\log\delta|^s$, $s>1$. So let $\Lambda_2^s=\left[-2\left \lfloor |\log \delta|^s \right \rfloor, 2\left \lfloor   |\log \delta|^s  \right \rfloor\right]$.
Denote by $S_2^s$ the probability  space such that there exist $\tilde{ \omega}^{\Lambda_2^s}$ satisfying \eqref{g18},   and  $m\in \Z^b $ with $|m|\leq|2\log\delta|^s$, either $|m_k|\geq 2$
for some $k\in \{1,2,\cdots,b\}$ or supp $m$ $\geq 3$,  and 
$j, j'\in \Lambda_2^s$ 
satisfying 
\begin{equation}\label{g24}
|m\cdot\tilde{ \omega}^{\Lambda_2^s}+\tilde{\mu}_j^{\Lambda_2^s}-\tilde{\mu}_{j'}^{\Lambda_2^s}|\leq4 \delta^{\frac{1}{8}},
\end{equation}
and the eigenvectors corresponding to $\tilde{\mu}_j^{\Lambda_2^s}$ and $\tilde{\mu}_{j'}^{\Lambda_2^s}$ satisfy that there exist  $\tilde{\ell}_{j}^{\Lambda_2^s}\in {\Lambda_2^s}$ and $ \tilde{\ell}_{j'}^{\Lambda_2^s}\in {\Lambda_2^s}$ such that
for any $\ell\in \Lambda_2^s$,
\begin{equation}\label{g25}
|\tilde{\phi}^{\Lambda_2^s}_{j}(\ell)|\leq C_0(1+|\tilde{\ell}_{j}^{\Lambda_2^s}|)^{q_0}e^{-\gamma_0|\ell-\tilde{\ell}_{j}^{\Lambda_2^s}|},|\tilde{\phi}^{\Lambda_2^s}_{j'}(\ell)|\leq C_0(1+|\tilde{\ell}_{j'}^{\Lambda_2^s}|)^{q_0}e^{-\gamma_0|\ell-\tilde{\ell}_{j'}^{\Lambda_2^s}|}.
\end{equation}

 We should mention that we allow $j=j'$.

\begin{theorem}\label{thm4}
	
	For small $\delta$, we have
 
		\begin{equation*}
		\mathbb{P}(	S_{2}^s )\leq (| \log \delta|)^{Cs} \delta^{\frac{1}{8}}.
		\end{equation*} 
 \end{theorem}

     \begin{proof}
     		Let  $N=2\left \lfloor |\log \delta|^s\right \rfloor$ and $\tilde{S}_2^s=S_{N}$. 
     	  By Lemma \ref{lem3} (Minami estimate), it suffices to prove that 
     	  	\begin{equation*}
     	  	\mathbb{P}(	S_{2}^s\cap \tilde{S}_2^s )\leq (| \log \delta|)^{Cs} \delta^{\frac{1}{8}}.
     	  \end{equation*} 
     	Denote by $P_{m,j,j'}$ the probability space such that  \eqref{g18} and \eqref{g24}  hold. 
     	Let  $m_i$ be such that $|m_i|=\max\{|m_k|, k=1,2,\cdots,b\}$. Without loss of generality, assume that $m_i=m_1$. Let us prove the case  $|m_1|\geq  2$ first. Without loss of generality, assume that $m_1\geq 0$.
     	It is easy to see 
     	that for any $\ell\in 2\tilde{B}_1$,
     	\begin{equation}\label{g83}
     	|\tilde{\phi}_{\tilde{\alpha}_k}(\ell)|^2\leq e^{-\gamma_0 L},\, k=2,3,\cdots, b.
     	\end{equation}
     	For any $l\in 2\tilde{B}_1$, denote by $P^3_l$ the probability space  such that 
     	\begin{equation}\label{g36}
     	|\tilde{\phi}_{\tilde{\alpha}_1} (l)|^2\geq \frac{4}{5}  |\tilde{\phi}_{j'} (l)|^2+  L^2(\sum_{k=2}^b\tilde{\phi}_{\tilde{\alpha}_k} (l)|^2 )+\frac{1}{L^{10}}.
     	\end{equation}	
     By \eqref{g72}, \eqref{g83} and \eqref{g36}, one has that
     	\begin{equation}\label{g79}
     		P_{m,j,j'}\subset \bigcup_{l\in2 \tilde{B}_1}P^3_l.
     	\end{equation}

    Now the proof follows from Lemma \ref{lemeig}, which  is similar to the proof of  Case $2_2$ or Case $3_1$.
    Here are the details.
     Take any $l\in 2 \tilde{B}_1 $ such that \eqref{g36} holds.
    Split $[0,1]$ into $N^{4q_1}$ intervals of size $N^{-4q_1}$ and  take any interval $I=[f_1,f_2]$. 
    Denote by the probability space $ P_{l, I}^3=\{V\in P_l: V_l\in I\}$.
    Applying  Lemma \ref{lemeig}, one has that if we fix $V_{\ell}$,  $ \ell\in\Lambda_1\backslash\{l\}$, then for any $f\in I$, 
    \begin{equation}\label{g91}
    	\frac{d\tilde{\mu}_{\tilde{\alpha}_1}^{V_l=f}}{df}=|\tilde{\phi}^{V_l=f_1}_{\tilde{\alpha}_1}(l)|^2+  N^{-5q_1}  O(1),
    \end{equation}
and
    \begin{equation}\label{g92}
    	\frac{d\tilde{\mu}_{j'}^{V_l=f}}{df}=|\tilde{\phi}^{V_l=f_1}_{j'}(l)|^2+  N^{-5q_1}  O(1).
    \end{equation}

    By \eqref{g91}, \eqref{g92} and \eqref{g93},  one has that
    \begin{align}
    	\frac{d (m\cdot\tilde{ \omega}^{V_l=f}+\tilde{\mu}_j^{V_l=f}-\tilde{\mu}_{j'}^{V_l=f})}{df}&\geq \frac{3}{4}|\tilde{\phi}^{V_l=f_1}_{\tilde{\alpha}_1}(l)|^2+  N^{-5q_1}  O(1)\nonumber\\
    	&	\geq \frac{3}{4L^{10}}+  N^{-5q_1}  O(1)\nonumber\\
    	&	\geq \frac{1}{2L^{10}}.\label{g94}
    \end{align}

    Now the proof follows that of  Case $2_2$ or Case $3_1$ of Theorem \ref{thm2}.
   
   Let us proceed to the case supp $m\geq 3$.
   	Without loss of generality, assume that $m_i\neq0,i=1,2,3$.
   In this case, if $\tilde{\ell}_j\geq 2L^4$, by \eqref{g25}, one has that for any $\ell\in\cup_{k=1}^b 2\tilde{B}_k $,
   \begin{equation}\label{g291}
   	|\tilde{\phi}_{j}(\ell)|\leq e^{- \gamma_0 L}.
   \end{equation}
   If $\tilde{\ell}_j\leq 2L^4$,   by \eqref{g30} and \eqref{g25}, one has that either for any $\ell\in 2\tilde{B}_1\cup2\tilde{B}_2$, 
   \begin{equation}\label{g391}
   	|\tilde{\phi} _{j}(\ell)|\leq e^{- \gamma_0 L},
   \end{equation}
   or 
   for any $\ell\in 2 \tilde{B}_2\cup 2\tilde{B}_3$, 
   \begin{equation}\label{g401}
   	|\tilde{\phi} _{j}(\ell)|\leq e^{- \gamma_0 L},
   \end{equation}
or
for any $\ell\in 2 \tilde{B}_1\cup 2\tilde{B}_3$, 
\begin{equation}\label{g402}
	|\tilde{\phi} _{j}(\ell)|\leq e^{- \gamma_0 L}.
\end{equation}
Clearly, \eqref{g291}-\eqref{g402} also hold for $j'$. Therefore, we have that there exists $ i\in\{1,2,3\}$, such that
for any $m\in\{\tilde{\alpha}_1,\tilde{\alpha}_2,\cdots,\tilde{\alpha}_b\}\backslash \{\tilde{\alpha}_i\}$ and $m=j,j'$, 
\begin{equation}\label{g4021}
	|\tilde{\phi} _{m}(\ell)|\leq e^{- \gamma_0 L},\ell\in 2\tilde{B}_i.
\end{equation}

   Now the proof follows that of Case $3_2$ of  Theorem \ref{thm2}.
   
    \end{proof}

Denote by $\hat{S}_2^s$ the probability  space such that there exist $\tilde{ \omega}^{\Lambda_2^s}$ satisfying \eqref{g18},  $m\in \Z^b $ with  $\sum_{k=1}^b m_k\neq 0$,  $|m_k|\leq 1, k=1,2,\cdots,b$,  and $j \in\Lambda_2^s, j'\in \Lambda_2^s$   satisfying
\begin{equation}\label{gg916}
	|m\cdot\tilde{ \omega}^{\Lambda_2^s}+\tilde{\mu}_j^{\Lambda_2^s}-\tilde{\mu}_{j'}^{\Lambda_2^s}|\leq4 \delta^{\frac{1}{8}}.
\end{equation}

\begin{theorem}\label{thm7}
	For small $\delta$,   we have
	\begin{equation*}
		\mathbb{P}(\hat{S}_2^s)\leq ( |\log \delta|)^{Cs}\delta^{\frac{1}{8}}.
	\end{equation*}

\end{theorem}

\begin{proof}
Again, let $N=2\left \lfloor |\log \delta|^s\right \rfloor$ and $\tilde{S}_2^s=S_{N}$. 
By Lemma \ref{lem3} (Minami estimate), it suffices to prove that 
	\begin{equation*}
	\mathbb{P}(\hat{S}_2^s\cap	\tilde{S}_{2}^s)\leq ( |\log \delta|)^{Cs}\delta^{\frac{1}{8}}.
\end{equation*} 
	Since  $\sum_{k=1}^b m_k\neq 0$, there exists $l\in [-N,N]$, such that 
	\begin{equation}\label{pos1}
		|(\sum_{k=1}^b m_k|\tilde\phi _{\tilde{\alpha}_k}(l)|^2)+|\tilde\phi_j (l)|^2-|\tilde\phi _{j'}(l)|^2|\geq \frac{1}{N^{10}}.
	\end{equation}
	
Now we use  Lemma \ref{lemeig} to conclude as in the proof of Case $2_2$ of Theorem \ref{thm2}. 
	
	
\end{proof}
 \begin{remark}\label{re7}
 	{\rm We could  replace the constant $4$ in \eqref{g5}, \eqref{g511}, \eqref{g24} and \eqref{gg916} with any fixed constant and all theorems in this section still hold. }
 \end{remark}

 \section{One dimensional random Schr\"odinger operators }
  We now proceed to the infinite volume random Schr\"odinger operator, $$H=-\Delta+V.$$ Let $\{\varphi^V_j\}_{j}$ be the eigen-basis 
  and assume that  $\iota_j^V$ satisfies 
 \begin{equation*}
  	|\varphi_j^V(\iota_j^V)|=\max_{x\in \Z} |\varphi_j^V(x)|.
  \end{equation*} 
 (If it is not unique, one may choose a maximum arbitrarily.)
  We have
  \begin{theorem}(See e.g., ~\cite{gelzx} or ~\cite[sect.~1.6]{GK14})\label{thm1}
  	There exist some $q>0$ and $\gamma_1>0$ such that, with probability 1, 
  	\begin{equation}\label{g1}
  	|\varphi^V_j(\ell)|\leq C_{V}(1+|\iota_j^V|)^qe^{-\gamma_1|\ell-\iota_j^V|},  
  	\end{equation} 
  	where  $\mathbb{E}(C_V)<\infty$.
  \end{theorem}
  \begin{remark} \label{label}
  	{\rm \begin{itemize}
  		\item Recall that $\iota_j^V$ is called the localization center.
  		\item  $\gamma_1$ and  $q$  only depend on the distribution $g$.   $\gamma_1$ can be arbitrarily close to the Lyapunov exponent. 
  	\end{itemize} }
  	 \end{remark}
  
  By Theorem \ref{thm1} and  following a similar proof of Theorem 7.1 in  ~\cite{djls}, one has the following Lemma concerning the localization centers.
  \begin{lemma}\label{lem11}
  	For any $\epsilon$, there exist  $\mathcal{V}_{\epsilon}$ with $  \mathbb{P}(\mathcal{V}_{\epsilon})>1-\epsilon$ and  a constant  $l_{\epsilon}$ such that  the following holds. 	
	For any $V\in \mathcal{V}_{\epsilon}$ 
  	  and $k\in[-L^4,L^4]$ with $L\geq l_{\epsilon}$, 
  		\begin{equation}\label{g13}
  		(1-\epsilon) L\leq \#\{j:\iota_j^V \in [k,k+L]  \}\leq (1+\epsilon) L.
  		\end{equation}
    \end{lemma}
 \noindent See appendix \ref{pf} for a proof. 	\hfill $\square$

 	Basing on \eqref{g13}, one may (re)label the eigenfunctions so that if
 	$j>j'$, then the localization centers  of the corresponding eigenfunctions $\phi_j^V$ and $\phi_{j'}^V$
 	satisfy  $\ell_{j'}^V\geq \ell_j^V$. The construction of such a map is presented in appendix \ref{label1}. Here after relabelling, we use the notations  $\phi_j^V$ instead of $\varphi_j^V$ and $\ell^V_j$ instead of $\iota_j^V$.  Recall that  $\mu_j^V$  is the   eigenvalue corresponding to eigenfunction $\phi_j^V$. When there is no ambiguity, we omit the dependence on $V$. 

  \smallskip
 
Below we summarize properties of the eigenfunction basis in this labelling. 
\begin{lemma}\label{lem1}
	There exist $q>0$ and $\gamma>0$ such that 
	for any $\epsilon$, there exist  $\mathcal{V}_{\epsilon}$ with $  \mathbb{P}(\mathcal{V}_{\epsilon})>1-\epsilon$ and constants $C_{\epsilon}$ and $\ell_{\epsilon}$ such that for any $V\in \mathcal{V}_{\epsilon}$, the following statements hold:
	\begin{itemize}
		\item  for any $\ell\in\Z$,
		\begin{equation}\label{g2}
			|{\phi}^V_j(\ell)|\leq C_{\epsilon}(1+|\ell_j^V|)^qe^{-\gamma|\ell-\ell_j^V|}, 
		\end{equation}
		\item for any $|\ell_j^V|\geq \ell_{\epsilon}$,
		\begin{equation}\label{g3}
			|\ell_j^V-j|\leq \epsilon |j|,
		\end{equation}
		\item for any $k\in [-L^4,L^4]$,  
	\begin{equation}\label{g321}
 (1- \epsilon) L\leq \#	\{j: \ell_j^V\in [k,k+L]\}\leq (1+\epsilon) L. 
	\end{equation}
	\end{itemize}
\end{lemma}
\begin{lemma}\label{lem4}
	Choose any $ V\in \mathcal{V}_{\epsilon}\cap {S}_{2N}$.  
	Let $\Lambda=[-2N,2N]$. Consider two    distinct  eigen-pairs $(\mu_j,\phi_j)$ and $(\mu_{j'},\phi_{j'})$, $|j|, |j'|\leq N$, of $H=-\Delta+V$.
		Then there exist  two distinct  eigen-pairs $( \tilde{\mu}_{\tilde{j}}^{\Lambda},\tilde{\phi}_{\tilde{j}}^{\Lambda})$  and
		$( \tilde{\mu}_{\tilde{j}'}^{\Lambda},\tilde{\phi}_{\tilde{j}'}^{\Lambda})$,     $\tilde{j}\in \Lambda$,  $\tilde{j}'\in \Lambda$ of $H_{\Lambda}$ such that 
	\begin{equation*}
		|\mu_j-\tilde{\mu}_{\tilde{j}}^{\Lambda}|\leq e^{-\frac{\gamma}{2} N},	|\mu_{j'}-\tilde{\mu}_{\tilde{j}'}^{\Lambda}|\leq e^{-\frac{\gamma}{2} N},
	\end{equation*}
	and
	\begin{equation*}
		||\phi_j-\tilde{\phi}_{\tilde{j}}^{\Lambda}||\leq  e^{-\frac{\gamma}{2} N},	||\phi_{j'}-\tilde{\phi}_{\tilde{j}'}^{\Lambda}||\leq e^{-\frac{\gamma}{2} N}.
	\end{equation*}
\end{lemma}
\begin{proof}
	By \eqref{g3}, one has that $|\ell_j|\leq (1+\epsilon) N$ and $|\ell_{j'}|\leq (1+\epsilon) N$. Then
	\begin{equation*}
		\sum_{|\ell|\geq 2N+1}|\phi_j(\ell)|^2\leq e^{-\frac{3\gamma}{2}N}, 	\sum_{|\ell|\geq 2N+1}|\phi_{j'}(\ell)|^2\leq e^{-\frac{3\gamma}{2}N},
	\end{equation*}
	and 
	\begin{equation}\label{gg4}
		||	H_{\Lambda}\phi_j-\mu_j\phi_j||\leq e^{-\frac{3\gamma}{4}N},||	H_{\Lambda}\phi_{j'}-\mu_j\phi_{j'}||\leq e^{-\frac{3\gamma}{4} N}.
	\end{equation}
	Therefore,  there exist  $\tilde{j}\in \Lambda$ and $\tilde{j}'\in \Lambda$ such that  
	\begin{equation*}
		|\mu_j-\tilde{\mu}_{\tilde{j}}^{\Lambda}|\leq e^{-\frac{3\gamma}{4}N},	|\mu_{j'}-\tilde{\mu}_{\tilde{j}'}^{\Lambda}|\leq e^{-\frac{3\gamma}{4} N}.
	\end{equation*}
	Since $V\in {S}_{2N}$, one has that  for any distinct $j_1$ and $j_2$ in $\Lambda$,
	\begin{equation*}
		|\tilde{\mu}^{\Lambda}_{{j}_1}-\tilde{\mu}_{j_2}^{\Lambda}|\geq \frac{1}{(2N)^{q_1}},
	\end{equation*}
	and hence for any $m\neq \tilde{j}$ (\text{or }$m\neq \tilde{j}'$), 
	\begin{equation}\label{gg7}
		|\tilde{\mu}^{\Lambda}_{m}-{\mu}_{{j}}|\geq \frac{1}{2^{q_1+1}N^{q_1}} \,( \text{or } |\tilde{\mu}^{\Lambda}_{m}-{\mu}_{{j}'}|\geq \frac{1}{2^{q_1+1}N^{q_1}}).
	\end{equation}
	Let $\tilde{\phi}_m$ and $\tilde{\mu}_m$  (as usual, for simplicity we have dropped the dependence on $\Lambda$ from $\tilde{\phi}_m^{\Lambda}$ and $\tilde{\mu}_m^{\Lambda}$, $m\in\Lambda$), be the eigen-pairs of $H_{\Lambda}$.
	Let
	\begin{equation}\label{gg6}
		I_{\Lambda}{\phi}_j=\sum_{m\in \Lambda} c_m\tilde{\phi}_m.
	\end{equation}
	From  \eqref{gg4} and \eqref{gg6}, one has that 
	\begin{equation}\label{gg10}
		H_{\Lambda} {\phi}_j = \sum_{m\in \Lambda}\tilde{\mu}_m c_m \tilde{\phi}_m,
	\end{equation}
	and 
	\begin{equation}\label{gg9}
		||\sum_{m\in \Lambda}\tilde{\mu}_m c_m \tilde{\phi}_m -\sum_{m\in \Lambda}{\mu}_j c_m \tilde{\phi}_m|| =O(1)e^{-\frac{3\gamma}{4} N}.
	\end{equation}

	By \eqref{gg7},  \eqref{gg10} and \eqref{gg9}, one has that for any  $m\neq \tilde{j}$, 
	\begin{equation*}
		|c_m| \leq  O(1)N^{q_1}e^{-\frac{3\gamma}{4} N}.
	\end{equation*}
	Therefore,  $ 1-O(1)N^{3q_1}e^{-\frac{3\gamma}{2} N}\leq c_{\tilde{j}}^2\leq 1$.
	We conclude that 
	\begin{equation*}
		||\phi_j-\tilde{\phi}_{\tilde{j}}|| \leq 	||\phi_j-I_{\Lambda}{\phi}_{{j}}|| +	||I_{\Lambda}{\phi}_{{j}}-\tilde{\phi}_{\tilde{j}}||  \leq  e^{-\frac{\gamma}{2} N}.
	\end{equation*}
	Similarly,
	\begin{equation*}
		||\phi_{j'}-\tilde{\phi}_{\tilde{j}'}|| \leq  e^{-\frac{\gamma}{2} N}.
	\end{equation*}
	Since $\phi_j$ and $\phi_{j'}$ are ortho-normal eigenfunctions, we have that $\tilde{j}\neq \tilde{j}'$.
	\end{proof}
 \smallskip
 
We now state the conclusion:
 
  \begin{theorem}\label{mainkeythm}
  	For any $\epsilon>0$, there exists $l_{\epsilon}$ such that the following statements hold.
  	Fix  any $L\geq \ell_\epsilon$ and $\beta_k\in\Z$, $k=1,2,\cdots,b$  satisfying 
  	$10L\leq |{\beta}_k|\leq L^3$ and $|{\beta}_{k}-{\beta}_{k'}|\geq  10L$, for any distinct $k,k'\in\{1,2,\cdots,b\}$,    there exists a probability space $X_{\epsilon}$ with  
  	$\mathbb{P}(X_\epsilon)\geq 1-\epsilon$ and   $\delta_0>0$ (depending on $g$, $\epsilon$ and  $L$) such that for any $V\in X_{\epsilon}$ and $0<\delta\leq \delta_0$,  
  	\begin{enumerate}
  		\item \begin{equation}\label{gg1}
  		|{\phi}_j(\ell)|\leq C_{\epsilon}(1+|\ell_j|)^qe^{-\gamma|\ell-\ell_j|}, 
  	\end{equation}
  	\item for any $|\ell_j|\geq l_{\epsilon}$,
  	\begin{equation}\label{gg2}
  		|\ell_j-j|\leq \epsilon |j|,
  	\end{equation}
  \item for  large $N$ (depending on $\epsilon$),  $|j|, |j'|\leq N$ and $j\neq j'$, 
  \begin{equation}\label{g61}
  	|{\mu}_j -{\mu}_{j'}|\geq \frac{1}{2^{q_1+1}N^{q_1}},
  \end{equation}
and 
 \begin{equation}\label{g611}
	|{\mu}_j|\geq \frac{1}{2N^{q_1}},
\end{equation}
\item  for any eigenfunction $\phi_{\alpha_k}$ with $\ell_{\alpha_k}\in B_k=\{l\in\Z: |l-{\beta}_k|\leq L\}$,  $k=1,2,\cdots,b$, we have that for any $(n, j)\in  [-e^{|\log\delta|^\frac{3}{4}},e^{|\log\delta|^\frac{3}{4}}]^{b+1}\backslash \{(- e_k, \alpha_k)\}_{k=1}^b$,
\begin{equation}\label{g51}
	|n\cdot{ \omega^{(0)}} +{\mu}_j |  \geq 2\delta^{\frac{1}{8}},
\end{equation}
where $\omega^{(0)}=(\omega^{(0)}_1,\cdots,\omega^{(0)}_b)=(\mu_{\alpha_1},\cdots,\mu_{\alpha_b})$,
and 
 for any $(n, j)\in  [-e^{|\log\delta|^\frac{3}{4}},e^{|\log\delta|^\frac{3}{4}}]^{b+1}\backslash \{(e_k, \alpha_k)\}_{k=1}^b$, 
\begin{equation}\label{g5111}
  	|-n\cdot{ \omega}^{(0)} +{\mu}_j |  \geq 2\delta^{\frac{1}{8}},
\end{equation}

\item  for any $\theta\in\R$, there are at most $b$ vertices 
$(n,j) \in [-|\log\delta|^s,|\log\delta|^s]^{b+1}$,   such that 
\begin{equation}\label{g52}
	|(n\cdot{ \omega^{(0)}}+\theta)+\mu_j|\leq \delta^{\frac{1}{8}},
\end{equation}
for any $\theta\in\R$, 
there are at most $b$ vertices
$(n,j) \in [-|\log\delta|^s,|\log\delta|^s]^{b+1}$,   such that 
\begin{equation}\label{g521}
	|-(n\cdot{ \omega^{(0)}}+\theta)+\mu_j|\leq \delta^{\frac{1}{8}}.
\end{equation}
 
  	\end{enumerate}
  \end{theorem}
  
%
\smallskip

\begin{proof}

By Lemmas \ref{lem2} (Wegner estimate), \ref{lem3} (Minami estimate), \ref{lem1}, \ref{lem4} and  Borel-Cantelli type arugments,, we have \eqref{gg1}-\eqref{g611}.

	
We apply the Theorems (with Remark \ref{re7}) in the previous section  with   $\tilde{B}_k=B_k$, $k=1,2,\cdots,b$ and $\delta=2^{-n}$, $n=1,2,\cdots$.  Then by  Borel-Cantelli type arugments, we have that there exists  	$ X_\epsilon $ with 	$\mathbb{P}(X_\epsilon)\geq 1-\epsilon$  and  $  X_\epsilon \cap (\tilde{S}_1\cup {S}_2^s \cup\hat{S}_2^s) =\emptyset$ for any small $\delta$.

	\eqref{g51} and \eqref{g5111} follow   from Lemma \ref{lem4} with $N=\left \lfloor e^{|\log\delta|^\frac{3}{4}} \right \rfloor$. 
	
	
	 The proof of \eqref{g52} and \eqref{g521} takes more time.  Without loss of generality, we only prove  
	 \eqref{g52}.  Assume  that there are   $({n}^{(m)},{j}^{(m)})\in 
	  [-|\log\delta|^s,|\log\delta|^s]^{b+1}$, $m=1,2,\cdots, b+1$,  satisfying
	  \begin{equation}\label{g121w1}
	  		|({n}^{(m)}\cdot{ \omega^{(0)}}+\theta)+\mu_{{j}^{(m)}}|\leq \delta^{\frac{1}{8}}.
	  \end{equation} 
  By Lemma  \ref{lem4}, one has that there exist $\tilde{ \omega}^{\Lambda_2^s}$ satisfying \eqref{g18}  and $(n^{(m)},\tilde{j}^{(m)})\in 
  [-|\log\delta|^s,|\log\delta|^s]^{b}\times [-2|\log\delta|^s,2|\log\delta|]^s$, $m=1,2,\cdots, b+1$  such that
  \begin{equation}\label{g121w2}
  	|(n^{(m)}\cdot \tilde{ \omega}^{\Lambda_2^s}+\theta)+\tilde{\mu}^{\Lambda_2^s}_{\tilde{j}^{(m)}}|\leq 2\delta^{\frac{1}{8}}.
  \end{equation} 
	  We can assume that $n^{(m)}$, $m=1,2,\cdots,b+1$ are distinct, otherwise  Lemma \ref{lem3} (Minami estimate) gives the proof.
	  
	  Choose any $m_1,m_2\in\{1,2,\cdots,b+1\}$ such that \eqref{g121w2} holds. 
	When $|n_k^{(m_1)}-n_k^{(m_2)}|\geq 2$ for some $k\in\{1,2,\cdots,b\}$ or supp $(n^{(m_1)}-n^{(m_2)})\geq 3$,  the proof follows from  Theorem \ref{thm4}.
	When $\sum_{k=1}^b (n_k^{(m_1)}-n_k^{(m_2)})\neq 0$ and $|n^{(m_1)}_k-n_k^{(m_2)}|\leq 1$, $k=1,2,\cdots,b$, the proof  follows from  Theorem \ref{thm7}.
	
	  	  So the only exceptional case is  when for  {\it all} $m_1,m_2\in\{1,2,\cdots,b+1\}$,  
		  $n^{(m_1)}$ and $n^{(m_2)}$ satisfy 
		  supp $(n^{(m_1)}-n^{(m_2)})=2$,  $\sum_{k=1}^b (n^{(m_1)}_k-n^{(m_2)}_k)=0$ and $n^{(m_1)}_k-n^{(m_2)}_k=\pm 1,0$, $k=1,2,\cdots,b$.  
	We will show that this is not possible. 
		  Shifting $n^{(m)}$ by $n^{(1)}$, one may assume that $n^{(1)}=(0,0,\cdots,0)$.  When $b=2$, it is obvious. 
		  So let $b\geq 3$. Without loss of generality, assume that
	  	  $n^{(2)}=(1,-1,0,\cdots,0)$. Thus either $n^{(m)}_1=1$ or $n^{(m)}_2=-1$ for all $m\in \{3,4,\cdots, b+1\}$.  Without loss of generality, assume that $n^{(3)}=(1,0,-1,0,0,\cdots,0)$.  Therefore, for all $m\in\{2,3, \cdots,b+1\}$,   $n_1^{(m)}=1$. This contradicts with $n^{(m)}$,  $m\in \{1, 2, \cdots, b+1\}$ being distinct.


		 \end{proof}

  \section{Large deviation theorem} \label{ldt}


Assume that $\tilde H$ is an operator on $\ell^2(\Z^{b+1}\times\{0,1\})$, T\"oplitz with respect to $n\in\Z^b$.  
We now write $\Z^{b+1}$ interchangeably with $\mathbb Z^b\times\mathbb Z$.
Assume that 
there exist functions $h_{r,r'}(n,j, j')$, $r, r'\in\{0,1\}$,
on $\mathbb Z^b\times\mathbb Z\times\mathbb Z$, such that for any $u_r(n,j)$, $r\in\{0,1\}$ and $(n,j)\in \Z^{b+1}$,

\begin{equation}\label{gatoe}
	\begin{aligned}
		( \tilde Hu)_r(n,j):
		&=\sum_{(n', j')\in \Z^b\times\Z, \, r'\in\{0,1\}} h_{r,r'}(n-n',j, j') u_{r'}(n',j').
	\end{aligned}
\end{equation}
Assume that there exist $C_1>0$ and $c_1>0$ such that 
\begin{equation}\label{gadecay}
	|h_{r,r'}(n,j, j')|\leq C_1e^{-c_1 (|n|+|j|+|j'|)}.
\end{equation}



Let ${D}(\theta)$ be a family of operators from $\R$  to Op[$\ell^2(\Z^{b+1}\times\{0,1\})$]:
\begin{equation}\label{glinearD}
	{D}(\theta)=\begin{bmatrix}D+
		&0\\
		0& D_- 
	\end{bmatrix}, 
\end{equation}
where $D_{\pm}=	\text{ diag}(\pm(n\cdot\omega+\theta)+\mu_j)$, $(n,j)\in \Z^{b+1}$.
Define $T=T(\theta)$: $\R\to \text{Op}[\ell^2(\Z^{b+1}\times \{0,1\})]$ as
\begin{equation}\label{gA}
	T(\theta)={D}(\theta)+\delta \tilde H.
\end{equation}

Denote by $Q_N$ an elementary region of size $N$ centered at 0, which is one of the following regions,
\begin{equation*}
	Q_N=[-N,N]^{b+1}
\end{equation*}
or
$$Q_N=[-N,N]^{b+1}\setminus\{n\in\mathbb{Z}^{b+1}: \ n_i\varsigma_i 0, 1\leq i\leq b+1\},$$
where  for $ i=1,2,\cdots,b+1$, $ \varsigma_i\in \{<,>,\emptyset\}$ and at least two $ \varsigma_i$  are not $\emptyset$.

Denote by $\mathcal{E}_N^{0}$ the set of all elementary regions of size $N$ centered at 0. Let $\mathcal{E}_N$ be the set of all translates of  elementary regions  with center at 0, namely,
$$\mathcal{E}_N:=\{n+Q_N:n\in\mathbb{Z}^{b+1},Q_N\in \mathcal{E}_N^{0}\}.$$
For simplicity, we call  elements in $\mathcal{E}_N$ elementary regions.  
Let $Q_N(j_0)=\{(n,j)\in\Z^b\times \Z: (n,j-j_0)\in Q_N\}$.  

The  width  of a   subset  $\Lambda\subset \Z^{b+1}$, is defined  as the maximum of 
$M\in \N$  such that  for any  $n\in \Lambda$, there exists  $\hat{M}\in \mathcal{E}_M$ such that
\begin{equation*}
	n\in \hat{M} \subset \Lambda
\end{equation*}
and
\begin{equation*}
	\text{ dist }(n,\Lambda\backslash \hat{M})\geq M/2.
\end{equation*}

A generalized  elementary region is defined to be a subset $\Lambda\subset \Z^{b+1}$ of the form
\begin{equation*}
	\Lambda:= R\backslash(R+z),
\end{equation*}
where $z\in\Z^{b+1}$ is arbitrary and $R$ is a rectangle,
\begin{equation*}
	R=\{(n_1,n_2,\cdots,n_{b+1})\in \Z^d: |n_1-n_1^\prime|\leq M_1, \cdots,|n_{b+1}-n_{b+1}^\prime|\leq M_{b+1}\}.
\end{equation*}
For $ \Lambda\subset\mathbb{Z}^{b+1}$,  we introduce its diameter,
$$\mathrm{diam}(\Lambda)=\sup_{n,n'\in \Lambda}|n-n'|.$$

Denote by $\mathcal{R}_N$
all  generalized elementary regions with diameters less than or equal  to $N$.
Denote by $\mathcal{R}_N^M$
all   generalized elementary regions in $\mathcal{R}_N$ with width larger than or equal to $M$.

With a slight abuse of notation, we also  use $\mathcal{E}_N$, $\mathcal{E}_N^{0}$, $Q_N$, $Q_N(j_0)$,  $\mathcal{R}_N$ and  $\mathcal{R}_N^M$  to denote $\mathcal{E}_N\times \{0,1\}$, $\mathcal{E}_N^{0}\times \{0,1\}$, $Q_N\times \{0,1\}$, $Q_N(j_0)\times \{0,1\}$,  $\mathcal{R}_N \times \{0,1\}$ and  $\mathcal{R}_N^M\times \{0,1\}$ 
respectively. 
Similarly for any $\Lambda\subset \Z^{b+1}$, denote by  $R_{\Lambda}$  the restriction  to $\Lambda\times\{0, 1\}$. 

We say that $T $ (given by \eqref{gA}) satisfies the large deviation theorem  (LDT) at scale $N$ with parameter $\tilde{c}_1$ if 
there exists a  subset $\Theta_N\subset \R$  such that
\begin{equation*}
	\mathrm{Leb}(\Theta_N)\leq e^{-{N}^{\frac{1}{30}}},
\end{equation*}
and for any $j_0\in [-2N,2N]$, $Q_N\in \mathcal{E}_N^{0}$,  and  $\theta \notin  \Theta_N  $,
\begin{equation}\label{gldt1}
	|| (R_{Q_N(j_0)}T(\theta)R_{Q_N(j_0)})^{-1} ||\leq e^{N^{\frac{9}{10}}}, 
\end{equation}
and  for  any $(n, j)$ and $(n', j')$ satisfying $ \max\{|n-n^\prime|,|j-j^\prime|\}\geq \frac{N}{10}$, 
\begin{equation} \label{gldt2}
	|(R_{Q_N(j_0)}T(\theta)R_{Q_N(j_0)})^{-1}(n,j;n^\prime,j')| \leq e^{-\tilde{c}_1\max\{|n-n^\prime|, |j-j'|\})}.
\end{equation} 	

Let $K_1$ be a large constant depending only on $b$. 
Let $K=K_1^{100},K_2=K_1^{5}$.
\begin{theorem}\label{thmldt}
	
	Assume that $\omega$ satisfies 
	\begin{enumerate}
		\item   $|\omega_k-\mu_{\alpha_k}|\leq C_2\delta$, $k=1,2,\cdots,b$;
		\item for any fixed  $N\geq (\log\frac{1}{\delta})^{K}$,
		and any  $\tilde{N}$  with $(\log\frac{1}{\delta})^{K}\leq \tilde{N}\leq N$  and $0\neq |n|\leq 2\tilde{N}$, 
		\begin{equation}\label{g1311}
			|n\cdot\omega| \geq e^{-\tilde{N}^{\frac{1}{K_2}}},
		\end{equation}
		and for any $ |j|\leq 3\tilde{N}, |j'|\leq 3\tilde{N}, |n|\leq 2\tilde{N}$ with $(n,j-j')\neq 0$,
		\begin{equation}\label{gg14}
			| n\cdot\omega-\mu_j+\mu_{j'}|\geq  e^{-\tilde{N}^{\frac{1}{K_2}}}.
		\end{equation}
	\end{enumerate}
	
	Then for small enough $\delta$, the LDT holds at any scale $\tilde{N}\leq N$ with parameter $\frac{1}{2}c_1$.
	
\end{theorem}
\begin{remark}{\rm
	\begin{enumerate}
		
		\item The proof of Theorem \ref{thmldt} uses ideas from the work of Bourgain, Goldstein and Schlag ~\cite{bgs}.  The implementation follows, however, 
		 the more recent paper ~\cite{liu}, which streamlined and quantified some of their arguments. 
		\item From \eqref{g1311} and \eqref{gg14}, in order to have LDT at all scales, we only need to remove measure (with respect to $\omega$) less than
		$e^{-\frac{1}{2}(\log \frac{1}{\delta})^{\frac{K}{K_2}}}\leq e^{-(\log \frac{1}{\delta})^{K_1^{90}}}\ll\delta$.
		\item  Theorem \ref{thmldt} holds for any parameter  $\tilde{c}_1$ with    $\tilde{c}_1<c_1$. 
	\end{enumerate}
	 }
\end{remark}
\subsection{Preparations}
\begin{lemma}~\cite[Prop. 14.1]{bbook}\label{mcl}
	Let $T(x)$ be a   $N\times N$ matrix function of a parameter $x\in[-\tau,\tau]$ satisfying the following conditions:
	\begin{itemize}
		\item[(i)] $T(x)$ is real analytic in $x\in [-\tau,\tau]$ and has a holomorphic extension to
		\begin{equation*}
			\mathcal{D}_{\tau,\tau_1}=\left\{z: |\Re z|\leq\tau,|\Im{z}|\leq \tau_1\right\}
		\end{equation*}
		satisfying
		\begin{equation}\label{mc1}
			\sup_{z\in \mathcal{D}_{\tau,\tau_1}}\|T(z)\|\leq B_1, B_1\geq 1.
		\end{equation}
		\item[(ii)]  For all $x\in[-\tau,\tau]$, there is a subset $\Lambda\subset [1,N]$ with
		\begin{equation*}|
			\Lambda|\leq M,
		\end{equation*}
		and
		\begin{equation}\label{mc2}
			\|(R_{[1,N]\setminus \Lambda}T(x)R_{[1,N]\setminus \Lambda})^{-1}\|\leq B_2, B_2\geq 1.
		\end{equation}
		\item[(iii)]
		\begin{equation}\label{mc3}
			\mathrm{Leb}\{x\in[-{\tau}, {\tau}]: \ \|T^{-1}(x)\|\geq B_3\}\leq 10^{-3}\tau_1(1+B_1)^{-1}(1+B_2)^{-1}.
		\end{equation}
		Let
		\begin{equation}\label{mc4}
			0<\epsilon\leq (1+B_1+B_2)^{-10 M}.
		\end{equation}
	\end{itemize}
	Then
	\begin{equation}\label{mc5}
		\mathrm{Leb}\left\{x\in\left[-{\tau}/{2}, {\tau}/{2}\right]:\  \|T^{-1}(x)\|\geq \epsilon^{-1}\right\}\leq C\tau e^{-c\left(\frac{\log \epsilon^{-1}}{M\log(B_1+B_2+B_3)}\right)},
	\end{equation}
	where $C$ and $c$ are absolute constants.
\end{lemma}

\smallskip



To apply Lemma \ref{mcl}, we also need to introduce semi-algebraic sets. A set $\mathcal{S}\subset \mathbb{R}^d$ is called {\it semi-algebraic} if it is a finite union of sets defined by a finite number of polynomial equalities and inequalities. More precisely, let $\{P_1,\cdots,P_s\}\subset\mathbb{R}[x_1,\cdots,x_d]$ be a family of real polynomials whose degrees are bounded by $\kappa$. A (closed) semi-algebraic set $\mathcal{S}$ is given by an expression
\begin{equation}\label{smd}
	\mathcal{S}=\bigcup\limits_{l}\bigcap\limits_{\ell\in\mathcal{L}_l}\left\{x\in\mathbb{R}^d: \ P_{\ell}(x)\varsigma_{l\ell}0\right\},
\end{equation}
where $\mathcal{L}_l\subset\{1,\cdots,s\}$ and $\varsigma_{l\ell}\in\{\geq,\leq,=\}$. Then we say that $\mathcal{S}$ has degree at most $s\kappa$. In fact, the degree of $\mathcal{S}$ which is denoted by $\deg(\mathcal{S})$, means the  smallest $s\kappa$ over all representations as in (\ref{smd}).

Following are some basic properties of these sets. They are special cases of that in ~\cite{ba}, and restated in ~\cite{bbook}.

\begin{lemma}~\cite[Theorem 9.3]{bbook} ~\cite[Theorem 1]{ba}\label{lediscom}
	Let $ \mathcal{S}\subset [0,1]^d$ be a semi-algebraic  set of degree $B$. Then  the number of connected components of
	$ \mathcal{S}$ does not exceed $(1+B)^{C(d)}$.
\end{lemma}

\begin{lemma}~\cite[Theorem 9.3]{bbook} ~\cite[Theorem 1]{ba}\label{lediscom1}
	Let $\mathcal{S}\subset [0,1]^{d_1+d_2}$ be a semi-algebraic  set of degree $B$. Let $(x,y)\in \R^{d_1}\times \R^{d_2}$.
	Then the projection ${\rm proj} _{x_1}\mathcal(S)$ is a semi-algebraic set of degree at most 
	$(1+B)^{C(d_1,d_2)}$.
\end{lemma}

\subsection{Large deviation theorem  for  small scales:  $N\leq (\log \frac{1}{\delta})^{10}$}

\begin{proof}
	In this case, let 
	\begin{align*}\tilde{ \Theta}_N=\{&\theta\in\R: \text{ there exists }  (n,j)\in [-N,N] ^{b}\times[-3N,3N] \text{ such that either }\\
		&|(n\cdot\omega +\theta)+\mu_j| \leq 2e^{-N^{\frac{1}{20}}}\text{ or } |(-n\cdot\omega -\theta)+\mu_j|  \leq 2e^{-N^{\frac{1}{20}}}\}.
	\end{align*}
	Clearly, ${\rm Leb} (\tilde{\Theta}_N)\leq  N^{C(b)}e^{-N^{\frac{1}{20}}}$.  
	When  $N\leq (\log \frac{1}{\delta})^{10}$, $\delta$ is much smaller than $e^{-N^{\frac{1}{20}}}$.
	Now \eqref{gldt1} and \eqref{gldt2} follow from  standard perturbation arguments. 
	
\end{proof}


\subsection{Large deviation theorem  for  intermediate scales: $(\log \frac{1}{\delta})^{10}\leq N\leq (\log \frac{1}{\delta})^K$}
\begin{proof}

	Let $\Theta_{N}$ be such that at least one of   \eqref{gldt1} and \eqref{gldt2} does not hold for $\theta\in \Theta_N$.
	Choose any $N\in [(\log \frac{1}{\delta})^{10}, (\log \frac{1}{\delta})^{K}] $.
	
	Assume that \eqref{gldt1} and  \eqref{gldt2} do not hold for some $\theta$. Then we must have  for some 
	$(n,j)\in[-N,N] ^{b}\times[-3N,3N] $, either $|\theta+n\cdot \omega^{(0)}+\mu_j|\leq C \delta$ or 
	$|\theta+n\cdot \omega^{(0)}-\mu_j|\leq C\delta$.
	Otherwise,  standard
	perturbation arguments yield that  for any  $j_0$ with $|j_0|\leq 2N$,
	\begin{equation*}
		|| (R_{Q_{N}(j_0)}TR_{Q_{N}(j_0)})^{-1} ||\leq\frac{1}{C\delta} \leq e^{N^{\frac{9}{10}}},
	\end{equation*}
	and for any $(n, j)$ and $(n', j')$ such that $\max\{|n-n^\prime|,|j-j^\prime|\}\geq \frac{N}{10}$, 
	\begin{equation*}
		|(R_{Q_{N}(j_0)}TR_{Q_{N}(j_0)})^{-1}(n,j;n^\prime,j')| \leq  e^{-\tilde{c}_1\max\{n-n^\prime|, |j-j'|\}|},
	\end{equation*}
	where $\tilde{c}_1$ can be any constant smaller than $c_1$.
	
	Therefore, we can restrict $\theta$ to be in $10^{b+1}N^{b+1}$ intervals of size $C\delta$. Denote all the intervals by $\{I_i\}$ and take one of them, $I_0$, into consideration. 
	Without loss of generality, assume that $I_0$ comes from the $+$ sector, namely
	\begin{equation*}
		I_0=\{\theta: |\theta+n_0\cdot \omega^{(0)}+\mu_{j_0}|\leq C {\delta} \text{ and } (n_0,j_0)\in [-N,N] ^{b}\times[-3N,3N]  \}.
	\end{equation*}
	For the  $- $ sector $ 		I_0=\{\theta: |\theta+n_0\cdot \omega^{(0)}-\mu_{j_0}|\leq C {\delta} \text{ and } (n_0,j_0)\in [-N,N] ^{b}\times[-3N,3N]  \},$ the proof is similar.

	Let   
	\begin{equation*}
		\mathcal{A}_1^\theta=\{(n,j)\in[-N,N] ^{b}\times[-3N,3N] : |\theta+n\cdot \omega^{(0)}+\mu_j|\leq \delta^{\frac{1}{8}}\},
	\end{equation*}
	and 
	\begin{equation*}
		\mathcal{A}_2^\theta=\{(n,j)\in[-N,N] ^{b}\times[-3N,3N] : |\theta+n\cdot \omega^{(0)}-\mu_j|\leq \delta^{\frac{1}{8}}\}.
	\end{equation*}
	By (5) of Theorem \ref{mainkeythm}, one has that for any $\theta$, 
	\begin{equation}\label{g1211}
		\# 	\mathcal{A}_1^\theta \leq b,  	\# 	\mathcal{A}_2^\theta\leq b.
	\end{equation}
	Since  the size of $I_0$ is $C\delta$, 
	we have that  there exist $	\mathcal{A}_1$ and $ 	\mathcal{A}_2$ independent of $\theta\in I_0$ such that  
	\begin{equation}\label{g121}
		\# 	\mathcal{A}_1\leq b,  	\# 	\mathcal{A}_2\leq b,
	\end{equation}
	and 
	for any $(n,j)\in[-N,N]^{b}\times[-3N,3N]\backslash(	\mathcal{A}_1\cup 	\mathcal{A}_2) $ and $\theta \in I_0$
	\begin{equation}\label{g120}
		|	\theta+ n\cdot \omega^{(0)}\pm\mu_j|\geq \frac{1}{2}\delta^{\frac{1}{8}}.
	\end{equation}
	Take any $\tilde{\Lambda}\in R_{N_1}^{\sqrt{N}}$ with $N_1\in[\sqrt{N}, 6N] $ and   $\tilde{\Lambda}\subset [-N,N] ^{b}\times[-3N,3N] $.
	By perturbation arguments, we have that   for any $\theta\in I_0$,
	\begin{equation}\label{gmc2}
		\|(R_{\tilde{\Lambda}\setminus(	\mathcal{A}_1\cup 	\mathcal{A}_2) }T(\theta)R_{\tilde{\Lambda}\setminus (	\mathcal{A}_1\cup 	\mathcal{A}_2) })^{-1}\|\leq 3 \delta^{-\frac{1}{8}}.
	\end{equation}

	We are going to apply Cartan's estimate, Lemma \ref{mcl}.
	For this reason, let  $\tau=C{\delta}$, $\tau_1=1$, $\Lambda=	\mathcal{A}_1\cup 	\mathcal{A}_2$, $M=2b$, $B_1=O(1)(\log \frac{1}{\delta})^K$, $B_2= 3 \delta^{-\frac{1}{8}}$,  $B_3=1$ and $\epsilon=e^{-{N_1}^{3/4}}$.
	We note that since $I_0$ has size $C\delta$, \eqref{mc3} holds automatically.
	Applying  Cartan's estimate (Lemma \ref{mcl})  in all possible $\tilde{\Lambda}\in R_{N_1}^{\sqrt{N}}$   (in total $N^C$), 
	there exists a  subset $\tilde{\Theta}_{N_1}\subset \R$  such that
	\begin{equation*}
		\mathrm{Leb}(\tilde{\Theta}_{N_1})\leq e^{-\frac{{N_1}^{\frac{3}{4}}}{|\log\delta|^2}},
	\end{equation*}
	and for any  $\theta \notin  \tilde{\Theta}_{N_1} $ and  any $\tilde{\Lambda}\in R_{N_1}^{\sqrt{N}}$  with $\tilde{\Lambda}\subset [-N,N] ^{b}\times[-3N,3N] $,
	\begin{equation}\label{g123}
		|| (R_{\tilde{\Lambda}}TR_{\tilde{\Lambda}})^{-1} ||\leq  e^{N_1^{\frac{3}{4}}}.
	\end{equation}
	Let $N_0=\sqrt{{N}}.$
	We call a box $(n_1,j_1)+Q_{N_0}\in \mathcal{E}_{N_0}$,  $(n_1,j_1)\in [-N,N] ^{b}\times[-3N,3N] $ good if 
	\begin{equation*}
		|| (R_{(n_1,j_1)+Q_{N_0}}TR_{(n_1,j_1)+Q_{N_0}})^{-1} ||\leq e^{N_0^{\frac{9}{10}}}
	\end{equation*}
	and for any $(n, j)$ and $(n', j')$ such that $\max\{|n-n^\prime|,|j-j'|\}\geq \frac{N_0}{10}$, 
	\begin{equation*}
		|(R_{(n_1,j_1)+Q_{N_0}}TR_{(n_1,j_1)+Q_{N_0}})^{-1} (n,j;n^\prime,j')| \leq e^{-\tilde{c}_1\max\{|n-n^\prime|,|j-j'|\}}.
	\end{equation*}
	Otherwise, we call $(n_1,j_1)+Q_{N_0}\in [-N,N] ^{b}\times[-3N,3N] $ bad. By \eqref{g121}, \eqref{g120} and perturbation arguments,
	we have that there are at most $2b$ disjoint bad boxes of size $N_0={N}^{1/2}$ contained in $ [-N,N] ^{b}\times[-3N,3N] $.
	
	We have sublinear bound and \eqref{g123}.   
	By  ~\cite[Theorem 2.1]{liu},  for any $\theta \notin  \bigcup_{\{I_i\}} \bigcup_{N_1\in[{\sqrt{N}}, 6{N}]}\tilde{\Theta}_{N_1}$, 
	\eqref{gldt1} and \eqref{gldt2} hold  for the scale ${N}$. Therefore, 
	\begin{equation*}
		\Theta_N\subset \bigcup_{\{I_i\}}   \bigcup_{N_1\in[\sqrt{{N}}, 6{N}]}\tilde{\Theta}_{N_1}
	\end{equation*}
	and hence
	\begin{equation*}
		{\rm Leb}	(\Theta_N)\leq    e^{-N^{\frac{1}{10}}}.
	\end{equation*}
	
	
\end{proof}

\subsection{Large deviation theorem  for large scales: $N\geq (\log \frac{1}{\delta})^K$}

\begin{proof}
	Let $N_2=N_1^{K_1}$ and  $N_4=N_2^{K_1}$. Assume  that $N_4\geq (\log \frac{1}{\delta})^K$ and that the LDT holds at both scales $N_1$ and $N_2$ with parameter $\tilde{c}_1$. 
	
	We will show that there are at most $N_1^C$ bad disjoint boxes of size $N_1$ contained in $[-N_4,N_4]^{b}\times[-3N_4,3N_4]$.
	Let $(n_1,l_1)\in [-N_4,N_4]^{b}\times[-3N_4,3N_4]$ be such that $(n_1,l_1)+Q_{N_1}$ is bad for some $Q_{N_1}\in \mathcal{E}^0_{N_1}$.
	
	We first bound the case when $|l_1|\leq 2N_1$. By the LDT at scale $N_1$, there exists a set $\Theta_{N_1}$ with ${\rm Leb} (\Theta_{N_1})\leq e^{-N^{1/30}_1}$ such that for any $\theta\notin \Theta_{N_1}$ and   any $Q_{N_1}\in \mathcal{E}^0_{N_1}$, $Q_{N_1}$ is good. Since the operator is T\"oplitz with respect to $n\in\Z^b$, one has that for any 
	$(n_1,l_1)$ with $\theta+n_1\cdot \omega\notin \Theta_{N_1}$,
	$(n_1,l_1)+Q_{N_1}$ is good for any $Q_{N_1}\in \mathcal{E}^0_{N_1}$.
	By standard arguments, we can assume that $\Theta_{N_1}$ is a semi-algebraic set of degree at most $N_1^C$, namely, there exist $N_1^C$ intervals $I_i$ of size  $e^{-N^{1/30}_1}$, 
	such that  $\Theta_{N_1}\subset \cup_i I_i$.
	The assumption on $\omega$ indicates that, for any nonzero $n$ with $|n|\leq 2N_4$, 
	\begin{equation}\label{g131}
		|n\cdot\omega| \geq e^{-(2N_4)^{\frac{1}{K_2}} }\geq e^{-N_1^{\frac{1}{30}}}.
	\end{equation}
	Therefore, for any  $|l_1|\leq 2N_1$,  there is at most one bad box  $ (n_1,l_1)$ such that   $n_1\cdot\omega \in I_i$.
	This leads to at most $N_1^C$ bad boxes in this case.
	%
	
	When $|l_1|\geq 2N_1$,
	we will show that there are at most three disjoint bad boxes of size $N_1$.  
	First, if a box $(n,j)+Q_{N_1}$ is bad, by \eqref{gadecay} and perturbation arguments, we must have that 
	for some $(n_1,l_1)\in (n,j)+Q_{N_1}$, either 
	\begin{equation}\label{gg12}
		| \theta+n_1\cdot\omega +\mu_{l_1}|\leq  2e^{-N_1^{9/10}}
	\end{equation}
	or
	\begin{equation}\label{gg13}
		|\theta+ n_1\cdot\omega -\mu_{l_1}|\leq  2e^{-N_1^{9/10}}.
	\end{equation}
	Assume that indeed there are three bad boxes.
	We have that  there are two from $D_+$, namely \eqref{gg12} (or $D_-$, namely \eqref{gg13}).  Therefore, we have that
	for two distinct vertices  $(n,j)\in [-N_4,N_4] ^{b}\times[-3N_4,3N_4] $ and $(n',j')\in [-N_4,N_4] ^{b}\times[-3N_4,3N_4] $, 
	\begin{equation*}
		| m\cdot\omega-\mu_j+\mu_{j'}|\leq  4e^{-N_1^{9/10}},m=n-n'.
	\end{equation*}
	This contradicts the assumption \eqref{gg14}.
	
	Let 
	$\tilde{\Theta}_{N_2}\subset \R$ be such that   for some  $(n,j)\in[-N_4,N_4] ^{b}\times[-3N_4,3N_4]  $ such that either  $|\theta+n\cdot\omega +\mu_j|\leq 2e^{-N_2^{9/10}}$ or 
	$|\theta+n\cdot\omega -\mu_j|\leq 2 e^{-N_2^{9/10}}$.
	Since for any $|l_1|\geq 2N_2$, the matrix is essentially diagonal, we have that for any $\theta\notin \tilde{\Theta}_{N_2}$, $(n_1,l_1)\in [-N_4,N_4] ^{b}\times[-3N_4,3N_4] $ with $|l_1|\geq 2N_2$ and $Q_{N_2}\in \mathcal{E}^0_{N_2}$,
	$(n_1,l_1)+Q_{N_2}$ is good. 
	Let 
	$\hat{\Theta}_{N_2}=\{\theta: \text{ for some } n\in [-N_4,N_4]^b, \theta+n\cdot \omega \in \Theta_{N_2}\}$.
	Therefore, for any 
	$\theta\notin \hat{\Theta}_{N_2}$, $(n_1,l_1)\in[-N_4,N_4] ^{b}\times[-2N_2,2N_2] $  and $Q_{N_2}\in \mathcal{E}^0_{N_2}$,
	$(n_1,l_1)+Q_{N_2}$ is good. 
	Clearly
	\begin{equation*}
		{\rm Leb} (\tilde{\Theta}_{N_2}\cap \hat{\Theta}_{N_2})\leq e^{-N_2^{1/31}},
	\end{equation*}
	and for any 
	$\theta\notin \tilde{\Theta}_{N_2}\cap \hat{\Theta}_{N_2}$, $(n_1,l_1)\in[-N_4,N_4] ^{b}\times[-3N_4,3N_4] $  and $Q_{N_2}\in \mathcal{E}^0_{N_2}$,
	$(n_1,l_1)+Q_{N_2}$ is good. 
	
	Applying Lemma \ref{mcl}
	(see proof of  ~\cite[Theorem 2.2]{liu} for details),   	for any $N_3\in [N_4^{1/2},N_4]$,   there exists a  subset $\tilde{\Theta}_{N_3}\subset \R$  such that 
	\begin{equation*}
		\mathrm{Leb}(\tilde{\Theta}_{N_3})\leq e^{-{N_3}^{\frac{1}{4}}},
	\end{equation*}
	and  for any 
	$N\in [N_3^{1/2},N_3]$, 
	$\tilde{\Lambda}\in R_{6N}^{N_3^{1/2} }  $ with $\tilde{\Lambda}\subset [-N_3,N_3] ^{b}\times[-3N_3,3N_3] $,
	and for any  $\theta \notin  \tilde{\Theta}_{N_3} $, 
	\begin{equation}\label{g1213}
		|| (R_{\tilde{\Lambda}}TR_{\tilde{\Lambda}})^{-1} ||\leq  e^{N^{\frac{3}{4}}}.
	\end{equation}
	Let $N_0= {{N_3}}^{1/2}.$
	We call a box $(n_1,l_1)+Q_{N_0}\in \mathcal{E}_{N_0}^{0}$,  $(n_1,l_1)\in [-{N}_3,{N}_3]^{b}\times[-3N_3,3N_3]$ good if 
	\begin{equation*}
		|| (R_{(n_1,l_1)+Q_{N_0}}TR_{(n_1,l_1)+Q_{N_0}})^{-1} ||\leq e^{N_0^{\frac{9}{10}}},
	\end{equation*}
	and for any $(n, j)$ and $(n', j')$ such that $\max\{|n-n^\prime|,|j-j'|\}\geq \frac{N_0}{10}$, 
	\begin{equation*}
		|(R_{(n_1,l_1)+Q_{N_0}}TR_{(n_1,l_1)+Q_{N_0}})^{-1} (n,j;n^\prime,j')| \leq e^{-c_2\max\{|n-n^\prime|,|j-j'|\}},
	\end{equation*}
where $c_2=\tilde{c}_1-N_4^{-\kappa}$ with a proper $\kappa>0$.
	Otherwise, we call $(n_1,l_1)+Q_{N_0}\in [-{N}_3,{N}_3]^{b}\times[-3N_3,3N_3]$ bad. 
	Since there are at most $N_1^C$ bad boxes of size $N_1$, by resolvent identity, 
	we have that there are at most $N_1^C$ disjoint bad boxes of size $N_0={N}_3^{1/2}$ contained in $ [-{N}_3,{N}_3]^{b}\times[-3N_3,3N_3]$.
	
	We have achieved the sublinear bound and \eqref{g1213}.   
	By ~\cite[Theorem 2.1]{liu},   we have that the LDT holds for any   scale $N_3\in [N_4^{1/2},N_4]$ with parameter $\tilde{c}_1-N_4^{-\kappa}$, where $\kappa$ is a proper small positive constant.   Now the proof follows from standard inductions.
\end{proof}
 \section{The nonlinear analysis}
  Fix $V\in X_\epsilon$, so that the conclusions of Theorem \ref{mainkeythm} hold and 
  Theorem \ref{thmldt} is available. (As before, we omit the 
  superscript $V$, as it is fixed.)
  We can now solve the nonlinear matrix equation \eqref{meq} and therfore the NLRS equation \eqref{g4}.
  \subsection{Lyapunov-Schmidt decomposition}
To simplify notations, 
   we write $u$ for $\hat u$, namely   $u(n, j)=\hat u(n, j)$.  Let $v$ be the complex conjugate of $u$,
   more precisely,   $v(n, j)=\overline{{\hat u}(-n, j)}$. 

By \eqref{W}, $W_u$ is a vector on $\ell^2(\Z^{b+1})$, which is now given by
  \begin{align}
   	W_u(n, j)=&\sum_{n'+\sum_{m=1}^p (n_m+n_m')=n\atop{n',n_m,n_m'\in\Z^b}} \sum_{l_m,l_m',j'\in\Z } {u}(n',j')\prod_{m=1}^p {u}(n_m,l_m) {v}(n'_m,l'_m)\nonumber\\
   	&\left(\sum_{x\in\Z}\phi_{j}(x)\phi_{j'}(x)\prod_{m=1}^p\phi_{l_m} (x) {\phi_{l_m'}(x)}\right).\label{W1}
   \end{align}
Let  $\widetilde{W}_u $  be  a vector on $\ell^2(\Z^{b+1})$, which is   given by
\begin{align}
	\widetilde{W}_u(n, j)=&\sum_{n'+\sum_{m=1}^p (n_m+n_m')=n\atop{n', n_m,n_m'\in\Z^b}} \sum_{j', l_m,l_m'\in\Z } {v}(n',j')\prod_{m=1}^p {v}(n_m,l_m) {u}(n'_m,l'_m)\nonumber\\
	&\left(\sum_{x\in\Z} \phi_{j}(x)\phi_{j'}(x)\prod_{m=1}^p\phi_{l_m} (x) {\phi_{l_m'}(x)}\right).\label{W2}
\end{align}
We remark that $ \widetilde{W}_u$ and ${W}_u$ are functions of $u$ and $v$. We only indicate  the dependence on $u$ for simplicity and the fact that $v$ is the conjugate of $u$.

   Writing the equation
   for $v$ as well, leads to the system of nonlinear equations on $\mathbb Z^{b+1}\times \{0, 1\}$:
   \begin{equation}\label{sys}
   \aligned
   &(D_+ u)(n,j)+\delta W_{u} (n,j)=0,\\
   &(D_- u)(n,j)+ \delta \widetilde{W}_u (n,j) =0,
   \endaligned
   \end{equation}
   where $D_\pm$ are the diagonal matrices with entries
\begin{equation}\label{D'}
D_\pm(n, j) :=D_\pm(n, j; n, j) = \pm n\cdot\omega+\mu_j.
\end{equation}

Define  
\begin{equation}
S=\{(-e_{k}, {\alpha_k})\times \{0\} ,  (e_{k}, {\alpha_k}) \times\{1\},\, k=1, 2, \cdots, b\},
\end{equation}
and denote the complement by $S^c$:
\begin{equation}
S^c=\mathbb Z^{b+1}\times \{0, 1\}\backslash S.
\end{equation}
Write $D$ for the diagonal matrix composed of the diagonal blocks $D_\pm$ and write \eqref{sys} in the form
$F(u,v)=0$. Since $v$ is the conjugate of $u$, we simply write $F(u)=0$ for $F(u,v)=0$. 
We make a Lyapunov-Schmidt decomposition of \eqref{sys} into the $P$-equations:
\begin{equation}\label{peq}
F(u)|_{S^c}=0.
\end{equation}
and the $Q$-equations:
\begin{equation}\label{qeq}
F(u)|_S=0.
\end{equation}

 \subsection{The $P$-equations}
 The $P$-equations are infinite dimensional. They are solved using a Newton scheme,
 starting from the initial approximation $u^{(0)}=u_0$.
Let $F'$ be the linearized operator
on $\ell^2(\mathbb Z^{b+1})\times \{0,1\}$,
$$F'(u)=D+\delta {\mathcal W}_u,$$
where $$\mathcal W_u= \begin{pmatrix} \frac{ \partial{W_u}}{\partial u}&  \frac{ \partial{W_u}}{\partial v}\\  \frac{ \partial{\widetilde{W}_u}}{\partial u} & \frac{ \partial{\widetilde{W}_u}}{\partial v}  \end{pmatrix}.$$

It is easy to see that $\mathcal W=\mathcal W_u$ satisfies

\begin{enumerate}
	\item $\mathcal W$ is T\"oplitz with respect to $n\in\Z^b$, namely for any $j\in\Z,j'\in\Z,k\in\Z^b,n\in\Z^b,n'\in\Z^b$, $r\in\{0,1\},r'\in\{0,1\}$
	\begin{equation*}
	 \mathcal W_{r,r'}(n, j; n', j')=\mathcal W_{r,r'}(n+k, j; n'+k, j').
	\end{equation*}
	\item  Assume that $|u(n,j)|\leq e^{-c(|n|+|j|)}$. Direct computation implies that 
	\begin{equation*}
		|\mathcal W_{r,r'}(n, j; n', j')|\leq Ce^{-c'(|n-n'|+|j|+|j'|)}.
	\end{equation*}
\end{enumerate}
The operator $F'$ is to be evaluated near $\omega=\omega^{(0)}=(\mu_{\alpha_1}, \mu_{\alpha_2}, \cdots, \mu_{\alpha_b})$, the linear frequency, and $u=u^{(0)}$ and $v=v^{(0)}$.
As earlier, we have that $u^{(0)}(-e_k,\alpha_k)=a_k$, $k=1,2,\cdots,b$; $u^{(0)}(n,j)=0$, otherwise.

Recall next the formal Newton scheme:
$$\Delta \begin{pmatrix} u\\v\end{pmatrix}=-[F'_{S^c}(u)]^{-1} F(u)|_{S^c},$$ 
where the left side denotes the correction to 
$\begin{pmatrix}u\\v\end{pmatrix}$, 
$F'_{S^c}(u)$ is the linearized operator evaluated at $(u, v)$ : $F'(u)$, and restricted to $S^c$: 
$$F'_{S^c}(u)(x, y)= F'(u)(x,y),$$ for $x, y\in S^c$; likewise $F(u)|_{S^c}$ is $F(u)$ restricted to ${S^c}$:
$$[F(u)|_{S^c}](x)=F(u)(x),$$ for $x\in S^c$. 

Since we seek solutions close to $(u^{(0)}, v^{(0)})$, which has compact support in $\mathbb Z^{b+1}\times \{0, 1\}$, 
we adopt a {\it multiscale} Newton scheme as follows:

\noindent At iteration step $(r+1)$, choose an appropriate scale $N$ and estimate $[F'_N]^{-1}$, 
where $F'_N$ is $F'$ restricted to 
$$[-N, N]^{b+1}\times \{0, 1\}\backslash S\subset \mathbb Z^{b+1}\times \{0, 1\},$$
and evaluated at $u^{(r)}$ and $v^{(r)}$:  $F'_N=F'_N(u^{(r)})$.
Define the $(r+1)$-th correction to be:  
$$\Delta \begin{pmatrix} u^{(r+1)}\\v^{(r+1)}\end{pmatrix}=-[F'_{N}(u^{(r)})]^{-1} F(u^{(r)}),$$ 
and 
$$\aligned u^{(r+1)}&=u^{(r)}+\Delta u^{(r+1)},\\
v^{(r+1)}&=v^{(r)}+\Delta v^{(r+1)},\endaligned$$
for all $r=0$, $1$, $2,\cdots$. 

At step $r$,  $\mathcal W_{u^{(r)} (\omega,a)}$ (depending on $u^{(r)}(\omega,a)$) is a function of $\omega$ and $a$.  We write $T_{u^{(r)}}(\theta,\omega,a)$ for  the operator 
$F'=D(\theta)+\delta \mathcal W_{u^{(r)}(\omega,a)}, $
and $\tilde{T}(\theta,\omega,a)$ the operator 
$F'=D(\theta)+\delta \mathcal W_{u^{(r)}(\omega,a)}, $  restricted to $\Z^{b+1}\times \{0,1\}\backslash S$,  where
\begin{equation}\label{glinearD2}
	{D}(\theta)=\begin{bmatrix}  \text {diag }(n\cdot\omega+\theta+\mu_j)
		&0\\
		0&  \text {diag }(-n\cdot\omega-\theta+\mu_j)
	\end{bmatrix},  (n,j)\in\Z^{b+1}.
\end{equation}
For simplicity, write  $\tilde{T}_{u^{(r)}}(\omega,a)$ for $\tilde{T}_{u^{(r)}}(0,\omega,a)$ and ${T}_{u^{(r)}}(\omega,a)$ for ${T}_{u^{(r)}}(0,\omega,a)$.

 The analysis of the linearized operators $F'_N$ uses Theorem \ref{mainkeythm} for small scales; for large scales, it
 also uses Theorem \ref{thmldt} and semi-algebraic projection to convert estimates
 in $\theta$ into that of $\omega$, and finally $a$. 
 
 \subsection{The $Q$-equations}
The $Q$-equations are $2b$ dimensional, but due to symmetry leading to $b$ equations only.
They are used to relate $\omega$ with $a$. The amplitudes $u(n, j)$ are  fixed on $S$, i.e., 
$u(-e_{\alpha_k}, \alpha_k)=a_k,\, k=1, 2, \cdots, b$, and the same for the complex conjugate.
So we have
\begin{equation}\label{omeq}
\omega_k=\mu_{\alpha_k}+\delta\frac{  W_u(-e_{k}, \alpha_k)}{a_k}, k=1, 2, \cdots, b.
\end{equation}

 When $u=u^{(0)}$, let us compute the terms in the $Q$-equations \eqref{omeq}. 
For $k\in\{1, 2, \cdots, b\}$, we have
\begin{align}
W_{u^{(0)}}(-e_k,\alpha_k)=\sum_{n'+\sum_{m=1}^p n_m-n'_m=-e_k} &u^{(0)}(n',l')\prod_{m=1}^p u^{(0)}(n_m,l_m)v^{(0)}(n'_m,l'_m) \label{WQ}\\
&\left(\sum_{x\in\Z} \phi_{\alpha_k}(x)\phi_{l'}(x)\prod_{m=1}^p\phi_{l_m}(x) \phi_{l_m'}(x)\right) \nonumber.
\end{align}
The sum in \eqref{WQ} runs over $l_m\in\Z$, $l_m'\in\Z$, $n'_m\in\Z^b$, $n_m\in\Z^b$, $ n'\in\Z^b$, $l'\in\Z$, $m=1,2,\cdots,p$.

Since $u^{(0)}$ has support $\{(-e_k,\alpha_k)\}_{k=1}^b$, in order to contribute to \eqref{WQ}, one has that
\begin{equation}\label{ggg1}
l_m\in\{\alpha_k\}_{k=1}^b,l'_m\in\{\alpha_k\}_{k=1}^b, m=1,2,\cdots,b \text{ and } l'  \in\{\alpha_k\}_{k=1}^b.
\end{equation}
Take $\sum_{x\in\Z}\phi_{\alpha_k}(x) \phi_{l'}(x)\prod_{m=1}^p\phi_{l_m}(x) \phi_{l_m'}(x) $ into consideration.
Assume  $l'=\alpha_k$ and  $ l_m=l_m'=\alpha_k$, $m=1,2,\cdots,b$. It is easy to see that  (similar to the proof of  \eqref{g72}), 
\begin{equation}\label{g72d}
\sum_{\ell\in\Z, |\ell-\ell_{\alpha_k}|\leq \frac{L}{2}} |{\phi}_{\alpha_k} (\ell)|^2\geq 1-e^{-\frac{\gamma}{4}L}.
\end{equation}
This implies that there exists $\ell\in\Z$ with $|\ell-\ell_{\alpha_k}|\leq \frac{L}{2} $ such that 
\begin{equation}\label{g72dd}
|{\phi}_{\alpha_k} (\ell)|\geq \frac{1}{L^3}.
\end{equation}
Therefore,  in this case, 
\begin{align}
\sum_{x\in\Z} \phi_{\alpha_k}(x)\phi_{l'}(x)\prod_{m=1}^p\phi_{l_m}(x) \phi_{l_m'}(x) &=\sum_{x\in\Z}| \phi_{\alpha_k}(x)|^{2p+2}\nonumber\\
&\geq \frac{1}{L^{10p}}.\label{ggg3}
\end{align}

Except for the case   $l'=\alpha_k$ and  $ l_m=l_m'=\alpha_k$, $m=1,2,\cdots,b$,  by \eqref{ggg1}, we have that 
\begin{equation}\label{ggg2}
|\sum_{x\in\Z} \phi_{\alpha_k}(x)\phi_{l'}(x)\prod_{m=1}^p\phi_{l_m}(x)\phi_{l_m'}(x) |\leq e^{-cL}.
\end{equation}

Denote by $A_k= \sum_{x\in\Z} |\phi_{\alpha_k}(x)|^{2p+2}$. 
By \eqref{WQ}, \eqref{ggg3} and \eqref{ggg2} (the leading contribution in the sum of \eqref{WQ} is when $(n',l')=(n'_m,l'_m)=(n_m,l_m)=(-e_k,\alpha_k)$), 
we have that 
\begin{equation*}
\omega^{(0)}_k=\mu_{\alpha_k}+\delta (A_k a_k^{2p}+O(1)e^{-cL}),
\end{equation*}
and $  \frac{1}{L^{10p}} \leq A_k\leq 1$.

Denote by $\Omega_0=[\mu_{\alpha_1} , \mu_{\alpha_1}+2^{2p+1}\delta]\times [\mu_{\alpha_2} , \mu_{\alpha_2}+2^{2p+1}\delta]\times [\mu_{\alpha_b} , \mu_{\alpha_b}+2^{2p+1}\delta]\subset\R^b$. 
Assume that after $r$ steps, we obtain a $C^1$ function $u^{(r)}(\omega, a)$ on $\Omega_0\times [1,2]^b$. 
Substituting $u^{(r)}(\omega, a)$ and $v^{(r)}(\omega, a)$ into \eqref{omeq}, the implicit function theorem yields,
 \begin{equation}\label{aeq}
\aligned
&\omega={\omega}^{(r+1)}(a),\\
 &a=a^{(r+1)}(\omega). 
 \endaligned
 \end{equation}
 Moreover,  for some  $C^1$ functions $f_k$, $k=1,2,\cdots,b$, 
\begin{equation}
\omega_k= \mu_{\alpha_k}+\delta (A_k a^{2p}_k+f_k(a_1,a_2,\cdots,a_b). 
\end{equation}
Denote by $\Gamma_r$, the graph of $(\omega, a)$ at step $r$.
Denote by $P_{x}$ the projection onto the $x$-variable, where $x=a$, $\omega$ or $(\omega,a)$.

\subsection{The induction hypothesis} 
 
Let $M$ be a large integer, and denote by $B(0, R)$ the
 $\ell^\infty$ ball on $\mathbb Z^{b+1}$ centered  at the origin with radius $R$.
 Set 
 $$r_0=\left\lfloor \frac{|\log\delta|^{\frac{3}{4}} } {\log M}\right\rfloor.$$
 
 The proof of the Theorem is an induction. So 
we first lay down the induction hypothesis, which we prove in sect.~6. 
In the following $C$ is a large constant (much larger than $K_2$,  which appears  in Theorem \ref{thmldt}) and $c>0$ is a small constant.

 For $r\geq 1$, we assume that the following holds:

\begin{itemize}
	\item[\bf Hi.] ${u}^{(r)}( \omega, a)$  is a $C^1$ map on  $\Omega_0\times [1,2]^b$, and
	$\text{supp } u^{(r)}\subseteq B(0, M^r)$ ($\text{supp } u^{(0)}\subset B(0, M)$).
	\item[\bf Hii.] $\Vert \Delta  u^{(r)}\Vert\leq \delta_r$,
	$\Vert \partial \Delta  u^{(r)}\Vert\leq \bar\delta_r$,
	where $\partial$ denotes $\partial_x$, $x$ stands for $\omega_i$, $a_i$, $i=1, 2, \cdots, b$
	and $\Vert\,\Vert:=\sup_{(\omega, a)}\Vert\,\Vert_{\ell^2(\Z^{b+1})}$.
	\item[\bf Hiii.] $|  {u^{(r)}}(n,j)|\leq  C  e^{ -c\max\{|n|,|j|\}}$.

	\item[\bf Hiv.] 
	There exists $\Lambda_r$, a set of open sets $I$ in $(\omega, a)$ of size $M^{-r^C}$ when $r\geq r_0$ (the total number of open sets is therefore bounded above by $M^{r^C}$),
	 such that   for any  $(\omega, a)\in \bigcup_{I\in \Lambda_r}I$ when $r\geq r_0$ and $(\omega, a)\in \Omega_0\times [1,2]^b$ when $1\leq r\leq r_0-1$,
	 	\begin{enumerate}
	 	

		\item  $u^{(r)}(\omega,a)$ is a  rational function in $(\omega, a)$  of degree at most $M^{r^3}$;
		\item 
		\begin{equation}\label{F}
		\Vert F( u^{(r)})\Vert\leq \kappa_r,\\
		\Vert \partial F( u^{(r)})\Vert\leq \bar\kappa_r; \end{equation}
		\item 
	         \begin{equation}\label{hgood1b}
		\Vert (R_{[-M^r,M^r]^{b+1}}\tilde{T}_{ u^{(r-1)}} (\omega,a)R_{[-M^r,M^r]^{b+1}})^{-1}\Vert \leq \delta^{-\frac{1}{8}}M^{r^C},
		\end{equation}
		and   for   $ \max\{|n-n^\prime|, |j-j'|\}>r^{C}$,
		\begin{equation}\label{hgood2b}
		| (R_{[-M^r,M^r]^{b+1}}\tilde{T}_{ u^{(r-1)}} (\omega,a)R_{[-M^r,M^r]^{b+1}})^{-1}(n,j;n^\prime, j')|\leq \delta^{-\frac{1}{8}}e^{-c\max\{|n-n'|,|j-j^\prime|\}}.
		\end{equation}
			\item 
				
				 Each $I\in {\Lambda}_r$ is contained in an open set $I'\in {\Lambda}_{r-1}$, $r\geq r_0$,
				and 
					\begin{equation}\label{gme1b1111}
					\text{Leb} (P_a(\Gamma_{r-1}\cap (\bigcup_{I'\in {\Lambda}_{r-1}} I'\backslash \bigcup_{I\in {\Lambda}_{r}}I)))\leq e^{-|\log \delta|^{ K_1^{90}}} +M^{-\frac{r_0}{2^b}}, r=r_0;
				\end{equation} 
			and 
				\begin{equation}\label{gme1b11}
				\text{Leb} (P_a(\Gamma_{r-1}\cap (\bigcup_{I'\in {\Lambda}_{r-1}} I'\backslash \bigcup_{I\in {\Lambda}_{r}}I)))\leq M^{-\frac{r}{2^b}}, r\geq r_0+1;
				\end{equation} 
					\item for  $(\omega, a)\in \bigcup_{I\in \Lambda_r}I$ with $r\geq r_0$, $\omega$  satisfies  the conditions \eqref{g1311}, and \eqref{gg14} for $n\neq 0$,
				in the scales $\tilde{N}$ in $[(\log\frac{1}{\delta})^K, M^{r}]$; 
			\item
		The iteration holds with 
		\begin{equation}\label{conv}
		\delta_r= \delta^{\frac{1}{2}}M^{-(\frac{4}{3})^r}, \, \bar\delta_r= \delta^{\frac{1}{8}} M^{-\frac{1}{2}(\frac{4}{3})^r}; \kappa_r= \delta^{\frac{3}{4}} M^{-(\frac{4}{3})^{r+2}}, \, \bar\kappa_r= \delta^{\frac{3}{8}} M^{-\frac{1}{2}(\frac{4}{3})^{r+2}}. \end{equation}
 
\end{enumerate}
\end{itemize}
\smallskip

 \begin{remark}
 	{\rm As usual in multi-scale arguments, the constant $c$ depends on $r$. 	From step $r$ to step $r+1$,    $c$  becomes slightly smaller. 
We ignore the  dependence  since it is essentially
irrelevant.}
 
 \end{remark}
 
  \begin{remark}
 	{\rm The Lyapunov-Schmidt approach to quasi-periodic solutions was initiated in the paper \cite{cw}, and greatly 
	generalized by Bourgain starting from the paper \cite{bann}.}
 
 \end{remark}

 \section{Proof of the Theorem}\label{proj}
   
    The general scheme of the proof is that, for small scales, we use \eqref{g51}  and \eqref{g5111} of Theorem~\ref{mainkeythm} to solve the $P$-equations; while 
for larger scales, we use \eqref{g52}  and \eqref{g521} of Theorem~\ref{mainkeythm}, Theorem~\ref{thmldt} and semi-algebraic projection.

 Let us state the projection lemma. 

\begin{lemma}\label{leproj}
	Let $\mathcal{S}\subset [0, 1]^{d_1}\times [0, 1]^{d_2}:=[0, 1]^{d}$, be a semi-algebraic set of degree $B$ and $\text{\rm meas}_{d} S <\eta, \log B\ll
	\log 1/\eta$.
	Denote by $(x, y)\in [0, 1]^{d_1}\times [0, 1]^{d_2}$ the product variable.
	Fix $\epsilon>\eta^{1/d}$.
	Then there is a decomposition
	$$
	\mathcal{S} =\mathcal{S}_1 \bigcup \mathcal{S}_2,
	$$
	with $\mathcal{S}_1$ satisfying
	$$
	\text{\rm Leb}({\rm Proj}_x \mathcal{S}_1)\leq B^C\epsilon,
	$$
	and $\mathcal{S}_2$ the transversality property
	$$
	\text{\rm Leb}(\mathcal{S}_2\cap L)\leq  B^C\epsilon^{-1} \eta^{1/d},
	$$
	for any $d_2$-dimensional hyperplane $L$ in  $[0, 1]^{d_1+d_2}$ such that
	$$
	\max_{1\leq l\leq d_1}|{\rm Proj}_L (e_l)|\leq  \frac 1{100}\epsilon,
	$$
	where $e_l$ are the basis vectors for the $x$-coordinates.
\end{lemma}

The above lemma is the basic tool,
underlining the semi-algebraic techniques used in the subject. 
It is stated as (1.5) in ~\cite{bgafa}, cf., Lemma~9.9 ~\cite {bbook} and
Proposition ~5.1 ~\cite{bgs}, and relies on the Yomdin-Gromov triangulation theorem. 
For a complete proof of the latter, see ~\cite{bin}. Together with Theorem \ref{thmldt},  \eqref{g52} and \eqref{g521}, it enables us to go beyond the
perturbative scales in \eqref{g51} and \eqref{g5111}. 

\begin{proof}[\bf Proof of induction hypothesis]
	Assume that the induction holds for all scales up to $r$. 	We will prove that  it holds for $r+1$. 
From our construction, it is easy to see that $u^{(r)}(\omega,a)$ is a  rational function in $(\omega, a)$  of degree at most $M^{(r+1)^3}$. 

For $r\leq r_0-1$,   by \eqref{g51}, \eqref{g5111} and standard perturbation arguments, we have that
for any $(\omega,a)\in\Omega_0\times[1,2]^b$, 
\begin{equation}\label{ggoodt1d1s1}
||(R_{[-M^{r+1},M^{r+1}]^{b+1}}\tilde{T} _{u^{(r)}} (\omega,a)R_{[-M^{r+1},M^{r+1}]^{b+1}})^{-1}||\leq  2\delta^{-\frac{1}{8}}M^{(r+1)^C},
\end{equation}
and for any  $(n, j)$ and $(n', j')$ satisfying $\max\{|n-n^\prime|,|j-j^\prime|\}\geq (r+1)^C$,
\begin{equation}\label{ggoodt2d2s1}
|(R_{[-M^{r+1},M^{r+1}]^{b+1}}\tilde{T}_{u^{(r)}} (\omega,a) R_{[-M^{r+1},M^{r+1}]^{b+1}})^{-1}(n,j;n^\prime,j^\prime)|\leq 2\delta^{-\frac{1}{8}}e^{- c \max\{|n-n^\prime|,|j-j^\prime|\}}.
\end{equation}

	We are in a position to treat the case $r\geq r_0$. 
	Let $\mathcal{X}\subset \cup_{I\in\Lambda_{r}} I$  be such that  
 $(\omega,a)\in \mathcal{X}$  satisfies
	   (5) of Hiv at step $r+1$ and
	     \begin{equation}\label{hgood1b11}
	   \Vert (R_{[-M^r,M^r]^{b+1}}\tilde{T}_{ u^{(r)}}  (\omega,a)R_{[-M^r,M^r]^{b+1}})^{-1}\Vert \leq \ \delta^{-\frac{1}{8}}M^{r^C},
	   \end{equation}
	   and   for  any $(n, j)$ and $(n', j')$ satisfying $\max\{|n-n^\prime|,  |j-j'|\}>r^{C}$,
	   \begin{equation}\label{hgood2b11}
	   | (R_{[-M^r,M^r]^{b+1}}\tilde{T}_{ u^{(r)}} (\omega,a)R_{[-M^r,M^r]^{b+1}})^{-1}(n,j;n^\prime, j')|\leq  \delta^{-\frac{1}{8}}e^{-c \max\{|n-n^\prime|,|j-j^\prime|\}}.
	   \end{equation}

	   By perturbation arguments,  
	   we can essentially assume that $\mathcal{X}$ is the union of a  collection of open intervals of  size  $M^{-(r+1)^C}$.   Denoting the collection by  $ \Lambda_{r+1}$,
	   we have constructed $\Lambda_{r+1}$.   Except for \eqref{gme1b1111} and  \eqref{gme1b11}, 
	it is now routine that the rest of Hi-v hold for $r+1$, see Chap.~18, IV, (18.36)-(18.41) ~\cite{bbook} and Lemmas~5.2
	and 5.3 ~\cite{wduke}. See  appendix \ref{ite} for more details.

We proceed to the proof of the measure estimates  \eqref{gme1b11} and \eqref{gme1b1111}.
Let $N=M^{r+1}$ and 
	$N_1=(\log N)^C$ with a large constant $C$.
	Let 
	\begin{equation*}
		r_1=\left\lfloor\frac{2\log N_1}{\log \frac{4}{3}}\right\rfloor+1,
	\end{equation*}
so that $\delta_{r_1}<e^{-N_1^2}$. 
	Consider $T_{u^{(r_1)}}$.  
	Pick one interval  $I\in\Lambda_{r_1}$ of size $M^{-r_1^C}$ and 
let $I_1=P_{\omega} (I\cap \Gamma_{r_1})$. By the $Q$-equation, we have that the size of $I_1$ is smaller than $C\delta$.

Solving the $Q$-equation at step $r_1-1 $, one has that $a=a^{(r_1)}(\omega)$, $\omega\in I_1$. 
Since $\omega\in\Omega_0$, one has that the first assumption of Theorem \ref{thmldt} always holds.
	By Theorem \ref{thmldt}, 
	there exists $X_{N_1}$ (depending on $\omega$) such that  for any $\theta\notin X_{N_1}$, 
	\begin{equation}\label{ggoodt1d}
	||(R_{Q_{N_1}}{T} _{u^{(r_1)}} (\theta,\omega,a^{(r_1)}(\omega))R_{Q_{N_1}})^{-1}||\leq   e^{N_1^{\frac{9}{10}}},
	\end{equation}
	and for any $(n,j)\in \Z^{b+1}$ and $ (n^\prime,j^\prime)\in\Z^{b+1}$ with $\max\{|n-n^\prime|,|j-j^\prime|\}\geq \frac{N_1}{10},$
	\begin{equation}\label{ggoodt2d}
	|(R_{Q_{N_1}}T_{u^{(r_1)}} (\theta,\omega,a^{(r_1)}(\omega)) R_{Q_{N_1}})^{-1}(n,j;n^\prime,j^\prime)|\leq  e^{- c \max\{|n-n^\prime|,|j-j^\prime|\}}, 
	\end{equation}
	and $X_{N_1} $ satisfies 
	 
	\begin{equation}\label{gsmalllineproj1d}
	{\rm Leb}( X_{N_1} )\leq e^{-N_1^{\frac{1}{30}}}.
	\end{equation}

 	Let $K_{N_1}=\{ \pm n\cdot \omega+\mu_j: |n|\leq N_1, |j|\leq 3N_1\}$ and $I_{N_1}$ (depending on $\omega$) be the $C\delta$ neighbour of $K_{N_1}$.
 Assume $\theta\notin I_{N_1}$.  Then the diagonal entries $D_+,D_-$ are larger than $C\delta$. Perturbation argument leads to,
 	for any $|j_0|\leq 	2N_1$ and $Q_{N_1}\in \mathcal{E}^0_{N_1}$, 
 	\begin{equation*}
 	\Vert (R_{Q_{N_1}(j_0)}T_{u^{(r_1)}}(\theta,\omega,a^{(r_1)}(\omega))R_{Q_{N_1}(j_0)})^{-1} \Vert\leq\frac{1}{\delta}, 
 	\end{equation*}
 	and  for any $(n, j)$ and $(n', j')$ satisfying  $ \max\{|n-n^\prime|,|j-j'|\}\geq \frac{N_1}{10}$,
 	\begin{equation*}
 	|(R_{Q_{N_1}(j_0)}T_{u^{(r_1)}}(\theta,\omega,a^{(r_1)}(\omega))Q_{Q_{N_1}(j_0)})^{-1}(n,n^\prime)| \leq  e^{-c \max\{|n-n^\prime|,|j-j^\prime|\}}.
 	\end{equation*}
 	
 Since  $N_1\geq (\log \frac{1}{\delta})^C$ ($C$ is large),  $\frac{1}{\delta}\leq e^{N_1^{\frac{1}{2}}}$. Therefore, \eqref{ggoodt1d} and \eqref{ggoodt2d} hold. This implies that
 	$X_{N_1}\subset I_{N_1}$. Since $\Omega_0$ has size $C\delta$, we can assume that $X_{N_1}$  is  in  a union of a collection of intervals of size $\delta$ with total number $N_1^C$ (independent of $\omega$).
	Pick one interval $\Theta$.
	  	Let $\mathcal{X}_{N_1} (\omega,\theta)\subset I_1\times \Theta$  be such that   there exists some $Q_{N_1}\in \mathcal{E}_{N_1}$  such that either \eqref{ggoodt1d} or \eqref{ggoodt2d} is not true.  
	  	By \eqref{gsmalllineproj1d} and Fubini theorem, one has that
	  	\begin{equation}
	  	{\rm Leb}(\mathcal{X}_{N_1} ) \leq C\delta e^{-N_1^{\frac{1}{30}}} \leq  \delta e^{-N_1^{\frac{1}{31}}}.
	  	\end{equation}
	We can assume that  
	$\mathcal{X}_{N_1} \subset I_1 \times\Theta$  is a semi-algebraic set of degree at most $N_1^CM^{Cr_1^3}$.  This can be seen as follows. Let $\tilde{X}_{N_1}\subset \Omega_0\times [1,2]^b\times \R$ be such that   
	there exists some $Q_{N_1}\in \mathcal{E}_{N_1}$  such that one of the following is not true:
		\begin{equation}\label{ggoodt1sd}
	||(R_{Q_{N_1}}{T} _{u^{(r_1)}} (\theta,\omega,a^{(r_1)}(\omega))R_{Q_{N_1}})^{-1}||\leq   e^{N_1^{\frac{9}{10}}},
	\end{equation}
	and for any $(n,j)\in \Z^{b+1}$ and $ (n^\prime,j^\prime)\in\Z^{b+1}$ with $\max\{|n-n^\prime|,|j-j^\prime|\}\geq \frac{N_1}{10},$
	\begin{equation}\label{ggoodt2sd}
	|(R_{Q_{N_1}}T_{u^{(r_1)}} (\theta,\omega,a^{(r_1)}(\omega)) R_{Q_{N_1}})^{-1}(n,j;n^\prime,j^\prime)|\leq  e^{- c \max\{|n-n^\prime|,|j-j^\prime|\}}.
	\end{equation}
Therefore, $\mathcal{X}_{N_1}=P_{(\omega,\theta)} (\tilde{X}_{N_1} \cap (\Gamma_{r_1}\times \R))$.  Clearly, both $\tilde{X}_{N_1} $ and $\Gamma_{r_1}$ are  semi-algebraic sets of degree at most $N_1^CM^{Cr_1^3}$.  
Lemma~\ref{lediscom1} implies that $\mathcal{X}_{N_1}$ is a semi-algebraic set of degree at most $N_1^CM^{Cr_1^3}$.

	Let
	\begin{equation}
	\epsilon_l=M^{-\frac{r}{2^{b-l}}},l=1,2,\cdots,b-1, \text{ and }\epsilon_b= 10M^{-r}.
	\end{equation}
	Choose any $|j_0|\leq  2N_1$. Recall that $T_{u^{(r_1)}}$ is T\"oplitz with respect to $n\in\Z^d$.  Denote by $\epsilon_{b+1}=e^{-N_1^{\frac{1}{40}}}$. 
	Applying  Lemma \ref{leproj}  in all possible directions  (see (3.26) in ~\cite{bgafa})  and also on all possible open sets  and $\Theta$ (the total number is bounded by $N_1^CM^{Cr_1^3}$), there exists a set of $\omega $, $I_1^r \subset I_1$ such that 
	\begin{equation}\label{gme4}
	{\rm Leb }(I_1^r)\leq \delta M^{Cr_1^3}N_1^C \left(\sum_{k=2}^{b+1} \left(\prod_{l=1}^{k-1}\epsilon^{-1}_l\right) \epsilon_k\right )\leq 
\delta	M^{-\frac{r}{2^{b-1}}}N_1^CM^{Cr_1^3},
	\end{equation}
	and for any  $\omega \in I_1\backslash    I_1^r$,
	one has that for any $(n_0,j_0)\in [-N,N]^{b}\times[-2N_1,2N_1]$ with $\max\{|n_0|,|j_0|\}\geq  \frac{M^r}{10}$ and $(\omega,a)\in \Gamma_{r_1}$,
	\begin{equation} 
	||(R_{(n_0,j_0)+Q_{N_1}}{T} _{u^{(r_1)}} (\omega,a)R_{(n_0,j_0)+Q_{N_1}})^{-1}||\leq  e^{N_1^{\frac{9}{10}}},
	\end{equation}
	and for any  $(n, j)$ and $(n', j')$ satisfying $\max\{|n-n^\prime|,|j-j^\prime|\}\geq \frac{N_1}{10},$
	\begin{equation}  
	|(R_{(n_0,j_0)+Q_{N_1}}T_{u^{(r_1)}} (\omega,a) R_{(n_0,j_0)+Q_{N_1}})^{-1}(n,j;n^\prime,j^\prime)|\leq  e^{- c \max\{|n-n^\prime|,|j-j^\prime|\}}. 
	\end{equation}
	Let us explain where the factor $\delta$ in \eqref{gme4} is from. 
	We apply Lemma  \ref{leproj}  in $I_1\times\Theta$, where both $I_1$ and $\Theta$ have sizes $C\delta$.  
	By scaling, we have such a $\delta$ factor.
	
	Assume $|j_0|> 2N_1$. In this case, we can assume that $R_{(n_0,j_0)+Q_{N_1}}{T} _{u^{(r_1)}} (\omega,a)R_{(n_0,j_0)+Q_{N_1}}$ is  essentially a diagonal matrix. So, we only need to remove $\omega$ such that for some $(n,j)\in[-N,N]^{b+1}$,
	$|n\cdot \omega+\mu_j|\leq 2e^{-N_1^{\frac{9}{10}}}$ or  	$|n\cdot \omega-\mu_j|\leq 2e^{-N_1^{\frac{9}{10}}}$.
This can not happen when $n=0$ because of \eqref{g611}. 
Direct compuations imply that there exists a set of $\omega $, $\tilde{I}_1^r \subset I_1$ such that 
\begin{equation}\label{gme41}
	{\rm Leb }(\tilde{I}_1^r))\leq N^{C(b)} e^{-N_1^{\frac{9}{10}}}< \delta	M^{-\frac{r}{2^{b-1}}},
\end{equation}
and for any  $\omega \in I_1\backslash    \tilde{I}_1^r$,
one has  that for any $(n_0,j_0)\in [-N,N]^{b+1}$ with $\max\{|n_0|,|j_0|\}\geq  \frac{M^r}{10}$, $|j_0|\geq 2N_1$ and $(\omega,a)\in \Gamma_{r_1}$,
\begin{equation} 
	||(R_{(n_0,j_0)+Q_{N_1}}{T} _{u^{(r_1)}} (\omega,a)R_{(n_0,j_0)+Q_{N_1}})^{-1}||\leq  e^{N_1^{\frac{9}{10}}},
\end{equation}
and for any $(n, j)$ and $(n', j')$ such that $\max\{|n-n^\prime|,|j-j^\prime|\}\geq \frac{N_1}{10},$
\begin{equation}  
	|(R_{(n_0,j_0)+Q_{N_1}}T_{u^{(r_1)}} (\omega,a) R_{(n_0,j_0)+Q_{N_1}})^{-1}(n,j;n^\prime,j^\prime)|\leq  e^{- c \max\{|n-n^\prime|,|j-j^\prime|\}}. 
\end{equation}

Since the distance between $\Gamma_{r_1}$ and $ \Gamma_{r}$ is less than $C\delta_{r_1}\leq  Ce^{-N_1^2}$ and $ ||u^{(r_1)}-u^{(r)}||\leq C\delta_{r_1}\leq Ce^{-N_1^2}$, 
	by perturbation arguments, for any $(n_0,j_0)\in [-M^{r+1},M^{r+1}]^{b+1}$ with $\max\{|n_0|,|j_0|\}\geq \frac{M^r}{10}$,  $(\omega,a)\in\Gamma_{r}$ and $\omega\in I_1\backslash (I_1^r\cup \tilde{I}_1^r)$,
	\begin{equation}\label{ggoodt1d1}
	||(R_{(n_0,j_0)+Q_{N_1}}{T} _{u^{(r)}} (\omega,a)R_{(n_0,j_0)+Q_{N_1}})^{-1}||\leq  2e^{N_1^{\frac{9}{10}}},
	\end{equation}
	and for any  $(n, j)$ and $(n', j')$ such that $\max\{|n-n^\prime|,|j-j^\prime|\}\geq \frac{N_1}{10},$
	\begin{equation}\label{ggoodt2d2}
	|(R_{(n_0,j_0)+Q_{N_1}}T_{u^{(r)}} (\omega,a) R_{(n_0,j_0)+Q_{N_1}})^{-1}(n,j;n^\prime,j^\prime)|\leq 2 e^{- c \max\{|n-n^\prime|,|j-j^\prime|\}}. 
	\end{equation}
	By  \eqref {hgood1b} and \eqref{hgood2b} at step $r$,   and using perturbation arguments, one has that
	for any $(\omega,a)\in\cup _{I\in \Lambda_{r}} I$,
	\begin{equation}\label{hgood1b1}
	\Vert (R_{[-M^r,M^r]^{b+1}}\tilde{T}_{u^{(r)}} (\omega,a)R_{[-M^r,M^r]^{b+1}})^{-1}\Vert \leq 2\delta^{-\frac{1}{8}}M^{r^C},
	\end{equation}
	and   for  any $(n, j)$ and $(n', j')$ such that $ \max\{|n-n^\prime|, |j-j'|\}>r^{C}$,
	\begin{equation}\label{hgood2b1}
	|(R_{[-M^r,M^r]^{b+1}}\tilde{T}_{u^{(r)}} (\omega,a)R_{[-M^r,M^r]^{b+1}})^{-1}(n,j;n^\prime, j')|\leq 2\delta^{-\frac{1}{8}}e^{-c \max\{|n-n^\prime|,|j-j^\prime|\}}.
	\end{equation}
	By  \eqref{ggoodt1d1}-\eqref{hgood2b1} and resolvent expansion as in Lemma~5.1 ~\cite{bw},
	one has that on $(\omega,a)\in ( \cup _{I\in \Lambda_{r}} I)\cap  \Gamma_{r}$, $\omega\in   (I_1\backslash (I_1^r\cup \tilde{I}_1^r))$
	\begin{equation}\label{ggoodt1d11}
	||(R_{[-M^{r+1},M^{r+1}]^{b+1}}\tilde{T}_{u^{(r)}} (\omega,a) R_{[-M^{r+1},M^{r+1}]^{b+1}})^{-1}||\leq  \delta^{-\frac{1}{8}}M^{(r+1)^C},
	\end{equation}
	and for  any $(n, j)$ and $(n', j')$ such that $\max\{|n-n^\prime|,|j-j^\prime|\}\geq (r+1)^C,$
	\begin{equation}\label{ggoodt2d21}
	|(R_{[-M^{r+1},M^{r+1}]^{b+1}}\tilde{T}_{u^{(r)}} (\omega,a) R_{[-M^{r+1},M^{r+1}]^{b+1}})^{-1}(n,j;n^\prime,j^\prime)|\leq   \delta^{-\frac{1}{8}}e^{- c \max\{|n-n^\prime|,|j-j^\prime|\}}.
	\end{equation}
In order to have (5) in Hiv, we have to remove $\omega $ of measure less than $ \delta^2M^{-\frac{r}{2}}$ from scales $M^{r} $ to $M^{(r+1)}$ for $r\geq r_0$, and of measure   less than $  \delta ^2 e^{-|\log \delta|^{K_1^{90}}}$ from scales $|\log \delta|^K$ to $M^{r_0}$.
		By counting all possible intervals $I$, and by \eqref{gme4}, \eqref{gme41},   one has that  for $r\geq r_0$, 
			\begin{equation}\label{gme1b111}
		\text{Leb} (P_{\omega}(\Gamma_{r}\cap (\bigcup_{I'\in {\Lambda}_{r}} I'\backslash \bigcup_{I\in {\Lambda}_{r+1}}I)))\leq \delta M^{r_1^C}M^{-\frac{r}{2^{b-1}}}+\delta^2M^{-\frac{r}{2}},
		\end{equation} 
	and 
	for $r= r_0-1$, 
	\begin{equation}\label{gme1b11111}
		\text{Leb} (P_{\omega}(\Gamma_{r}\cap (\bigcup_{I'\in {\Lambda}_{r}} I'\backslash \bigcup_{I\in {\Lambda}_{r+1}}I)))\leq \delta M^{r_1^C}M^{-\frac{r}{2^{b-1}}}+ \delta^2 e^{-|\log \delta|^{K_1^{90}}}.
	\end{equation} 
		This implies that  for $r\geq r_0$
			\begin{equation}\label{gme1b112}
		\text{Leb} (P_{a}(\Gamma_{r}\cap (\bigcup_{I'\in {\Lambda}_{r}} I'\backslash \bigcup_{I\in {\Lambda}_{r+1}}I)))\leq M^{-\frac{r}{2^b}},
		\end{equation} 
	and 
	 for $r= r_0-1$
	\begin{equation}\label{gme1b1121}
		\text{Leb} (P_{a}(\Gamma_{r}\cap (\bigcup_{I'\in {\Lambda}_{r}} I'\backslash \bigcup_{I\in {\Lambda}_{r+1}}I)))\leq M^{-\frac{r}{2^b}} +e^{-|\log \delta|^{K_1^{90}}}.
	\end{equation} 
We have completed the proof.

\end{proof}

 \begin{proof}[\bf Proof of main theorem]
 	Theorem \ref{mainthm} follows immediately from the (by now verified) hypothesis (Hi-iv). 
 \end{proof}
 
 \appendix

 \section{Eigenfunction relabelling map} \label{label1}
 For a fixed $V\in\mathcal V_\epsilon$, basing on \eqref{g13}, one may relabel the eigenfunctions in a more intrinsic way. 
  Write $\varphi_i$ and $\iota_i$ for $\varphi_i^V$ and $\iota_i^V$, since $V$ is fixed. 
  The goal is that in the new labelling scheme, if $j>j'$, then the localization centers  of the corresponding eigenfunctions $\phi_j$ and $\phi_{j'}$
  satisfy  $\ell_{j'}\geq \ell_j$.
  Below we provide such a relabelling map. 
  
  For a given eigenfunction $\varphi_i$, we first select a vertex
  among the set of vertices, on which $\varphi_i$ achieves its maximum.
  (This selection could be arbitrary, but it is practical to have a rule.) 
  So for $i\in\mathbb Z$, 
  let $$\mathcal M_i=\{x_0\in\mathbb Z:|\varphi_i(x_0)|=\max_{x\in\mathbb Z}|\varphi_i(x)|\}.$$
  Define $\mathcal M_i^+=\mathcal M_i\cap\{\{0\}\cup\mathbb Z_+\}.$
  If $\mathcal M_i^+\neq\emptyset$, define
  $\iota_i=\min x_0$, $x_0\in\mathcal M_i^+$; otherwise define
  $\iota_i=\min -x_0$, $x_0\in \mathcal M_i.$
  
  Define 
  $$f_1: \mathbb Z\mapsto \mathbb Z, \, f_1(i)=\iota_i.$$
  Let $\mathcal L$ be the range of $f_1$,  
  $$\mathcal L=\text {Ran}(f_1), \mathcal L\subseteq \mathbb Z.$$
  For a given $l\in\mathcal L$, let $I_l=\{i|\iota_i=l\}$. Define
  $$f_2: \mathcal L\mapsto  \bigsqcup_{l\in\mathcal L} I_l.$$
  From \eqref{g13}:
  If $|l|<l_\epsilon$, $\#\{ \cup_{|l|\leq l_{\epsilon}} I_l\}\leq (1+\epsilon)l_\epsilon$;
  and if $|l|\geq l_\epsilon$, 
   $$(1-\epsilon)l\leq\#\{ \cup_{|l|\leq l_{\epsilon}} I_l\}\leq (1+\epsilon)l.$$
  The upper bound gives that
  for all $l$, $I_l$ is finite. So we may define a map
  $$f_3: \bigsqcup_{l\in\mathcal L} I_l\mapsto \mathbb Z,$$
  such that if $x\in I_l$ and $y\in I_{l'}$, with $l'>l$,
  then $f_3(y)>f_3(x)$.

  
  Finally define the map $f$ to be 
   $f=f_3\circ f_2\circ f_1$,
    $$f: \mathbb Z \mapsto \mathbb Z.$$
Using the relabelling map $f$ yields our ortho-normal eigen-basis $\{\phi_j\}_{j\in\mathbb Z}$.
   \smallskip

 \section{Proof of Lemma \ref{lem11}}\label{pf}
 
 \begin{proof}
	By Theorem \ref{thm1} and Chebyshev's inequality, one has that 
	for any $\epsilon_1$, there exists  $\mathcal{V}_{\epsilon_1}$ with $\mathbb P(\mathcal{V}_{\epsilon_1})>1-\epsilon_1$  such that for any $V\in \mathcal{V}_{\epsilon_1}$, 
	\begin{equation}\label{appg5}
		|\varphi^V_j(\ell)|\leq C_{\epsilon_1}(1+|\iota_j^V|)^qe^{-\gamma_1|\ell-\iota_j^V|}.
	\end{equation}
	For simplicity, below we drop the superscript $V$. 
	Clearly, we have
	\begin{equation}\label{appg1}
		\sum_{\ell\in\Z} |\varphi_j(\ell)|^2=1.
	\end{equation}
	and
	\begin{equation}\label{appg2}
		\sum_{j\in\Z} |\varphi_j(\ell)|^2=1.
	\end{equation}
Let $\epsilon$ be an arbitrarily small constant.
	Assume $L$ is large enough, depending on $\epsilon$ and $\epsilon_1$.
	If $k\leq \iota_j\leq k+L$ with $|k|\leq L^4$, then
	\begin{align}
		\sum_{\ell \leq k -\epsilon L} |\varphi_{j}(\ell)| ^2+\sum_{\ell \geq k+(1+\epsilon) L} |\varphi_{j}(\ell)| ^2&\leq \sum_{m\geq \epsilon L} C_{\epsilon_1}  L^{8q}e^{-\gamma_1 m} \\
		&\leq Ce^{-\epsilon L},\label{appg3}
	\end{align}
	where $C$ depends on $\epsilon_1$ and $\epsilon$.
	By \eqref{appg1} and \eqref{appg3},  one has that for any $j$ with $k \leq \iota_j\leq k+ L$,
	\begin{equation}\label{appg4}
		\sum_{k -\epsilon L\leq  \ell \leq k+(1+\epsilon)L} |\varphi_{j}(\ell)| ^2\geq 1- C e^{-\epsilon L}.
	\end{equation}
	By \eqref{appg2} and \eqref{appg4}, we have that
	\begin{align}
		(1+\epsilon)L&\geq 	\sum_{k -\epsilon L\leq  \ell \leq k+(1+\epsilon)L\atop^{j\in\Z}} |\varphi_{j}(\ell)| ^2\\
		&\geq 	\sum_{k-\epsilon L\leq  \ell \leq k+ (1+\epsilon)L\atop^{j\in\Z:k \leq \iota_j\leq k+ L}} |\varphi_{j}(\ell)| ^2\\
		&\geq (1- Ce^{-\epsilon L}) \#\{j:k\leq  \iota_j\leq k+L\} .
	\end{align}
	This implies that for any $k\in[-L^4,L^4]$,
	\begin{equation}\label{appg9}
		\#\{j:k \leq \iota_j\leq k+L\} \leq (1+\epsilon ) L.
	\end{equation}
	For any $\ell\in[ k+\epsilon L, k+(1-\epsilon) L]$ with $|k|\leq L^4$,  by \eqref{appg5}, one has
	\begin{align}
		\sum_{j\in\Z:\iota_j \notin [k,k+ L]}  |\varphi_{j}(\ell)| ^2&\leq \sum_{m=0}^{\infty}	\sum_{j\in\Z:m  L \leq |\iota_j|\leq (m+1) L\atop^{|\iota_j-\ell| \geq \epsilon L}} C(1+|\iota_j|)^{2q}e^{-\gamma_1|\ell-\iota_j|}\nonumber\\
		&\leq \sum_{m\leq 10L^4}	\sum_{j\in\Z:  |\iota_j|\leq  20 L^5} CL^{10q}e^{-\gamma_1\epsilon L}+\nonumber\\
		& \sum_{m=10 L^4}^{\infty}  C(1+m L)^{2q}	\#\{j:  |\iota_j|\leq (m+1)L\} e^{-\frac{1}{2}\gamma_1 mL} \nonumber\\
		&\leq C e^{-  \epsilon L} + \sum_{m=1}^{\infty}  C(1+m L)^{3q} e^{-\frac{1}{2}\gamma_1 mL}\label{appg11}\\
		&\leq  C e^{-\epsilon L},\nonumber
	\end{align}
	where   \eqref{appg11} holds by \eqref{appg9}. 
	It implies that
	\begin{equation}\label{appg8}
		\sum_{j\in\Z:\iota_j \notin [k,k+ L] \atop^{\ell\in[k+\epsilon L, k+ (1-\epsilon) L]}}  |\varphi_{j}(\ell)| ^2\leq   C e^{-\epsilon L}.
	\end{equation}
	By \eqref{appg2} and \eqref{appg8}, one has that
	\begin{align*}
		(1-\epsilon)L&=	\sum_{\ell\in[k+\epsilon L, k+(1-\epsilon)L],j\in\Z} |\varphi_{j}(\ell)| ^2\\
		&\leq Ce^{-\epsilon L}+	\sum_{j\in\Z:\iota_j \in [k,k+L] \atop^{\ell\in[k+\epsilon L, k+(1-\epsilon) L]}}  |\varphi_{j}(\ell)| ^2 \\
		&\leq C e^{-\epsilon L}+ \#\{j:k\leq \iota_j\leq k+L\} .
	\end{align*}
	This yields that  for any $k$ with $|k|\leq L^4$,
	\begin{equation}\label{appg13}
		\#\{j:k\leq \iota_j\leq k+ L\} \geq (1- \epsilon ) L.
	\end{equation}
	
	Now \eqref{g13} follows from    \eqref{appg9} and \eqref{appg13}.
\end{proof}

 \section{Bourgain's  induction estimates}\label{ite}
   Assume that Hi-v hold at step $r$, and that  Hi-v hold at step  $r+1$ except for \eqref{conv}. We show that 
    \eqref{conv} holds at step $r+1$.
    This follows from  Chap.~18 ~\cite{bbook} with minor modifications.
 
 Let us state the difference between our setting and that of Bourgain's.
 Bourgain assumed that 
 \begin{equation} 
 \Vert (R_{[-M^r,M^r]^{b+1}}\tilde{T}_{ u^{(r-1)}} R_{[-M^r,M^r]^{b+1}})^{-1}\Vert \leq  M^{r^C},
 \end{equation}
 and   that for  $(n, j)$ and $(n', j')$ satisfying $\max\{ |n-n^\prime|, |j-j'|\}>r^{C}$,
 \begin{equation} 
 | (R_{[-M^r,M^r]^{b+1}}\tilde{T}_{ u^{(r-1)}} R_{[-M^r,M^r]^{b+1}})^{-1}(n,j;n^\prime, j')|\leq e^{-c\max\{|n-n'|,|j-j^\prime|\}}.
 \end{equation}
 Then he proved that the induction holds if 
 \begin{align*}
 \delta_{r+1}&\geq  M^{(r+1)^C}\kappa_r,\\
 \bar{\delta}_{r+1}&\geq M^{2(r+1)^C}\bar{\kappa}_r+M^{(r+1)^C}\delta_{r+1},\\
 \kappa_{r+1}&\geq e^{-\frac{c}{3}M^{r+1}}\kappa_r+\delta_{r+1}^2,\\
 \bar{\kappa}_{r+1}&\geq  M^{2(r+1)^C}{\kappa}_r +  e^{-\frac{c}{3}M^{r+1}}\bar{\kappa}_r+\delta_{r+1}\bar{\delta}_{r+1}.
 \end{align*}
 
 We  assume  that for a proper $0<\nu<1$,
\begin{equation} 
\Vert (R_{[-M^r,M^r]^{b+1}}\tilde{T}_{ u^{(r-1)}} R_{[-M^r,M^r]^{b+1}})^{-1}\Vert \leq \delta^{-\nu}M^{r^C},
\end{equation}
and  that  for  $(n, j)$ and $(n', j')$ satisfying $\max\{|n-n^\prime|, |j-j'|\}>r^{C}$,
\begin{equation} 
| (R_{[-M^r,M^r]^{b+1}}\tilde{T}_{ u^{(r-1)}} R_{[-M^r,M^r]^{b+1}})^{-1}(n,j;n^\prime, j')|\leq \delta^{-\nu}e^{-c\max\{|n-n'|+|j-j^\prime|\}}.
\end{equation}
Following Bourgain's proof, we obtain the following new relations
\begin{align*}
\delta_{r+1}&\geq\delta^{-\nu} M^{(r+1)^C}\kappa_r,\\
\bar{\delta}_{r+1}&\geq \delta^{-2\nu} M^{2(r+1)^C}\bar{\kappa}_r+\delta^{-\nu} M^{(r+1)^C}\delta_{r+1},\\
\kappa_{r+1}&\geq \delta^{1-\nu} e^{-\frac{c}{3}M^{r+1}}\kappa_r+\delta_{r+1}^2,\\
\bar{\kappa}_{r+1}&\geq \delta^{-2\nu} M^{2(r+1)^C}{\kappa}_r +\delta^{1-\nu} e^{-\frac{c}{3}M^{r+1}}\bar{\kappa}_r+\delta_{r+1}\bar{\delta}_{r+1}.
\end{align*}
 For example, we may take $\nu=\frac{1}{8}$
\begin{equation*}
\delta_r=\delta^{\frac{1}{2}} M^{-(\frac{4}{3})^r}, \bar{\delta}_r=\delta^{\frac{1}{8}} M^{-\frac{1}{2}(\frac{4}{3})^r}, \kappa_r=\delta^{\frac{3}{4}} M^{-(\frac{4}{3})^{r+2}}, \bar{\kappa}_r=\delta^{\frac{3}{8}} M^{-\frac{1}{2}(\frac{4}{3})^{r+2}}.
\end{equation*}

	\section*{Acknowledgments}

W. Liu was supported by NSF DMS-2000345 and DMS-2052572. W.-M.~Wang acknowledges support from CY Initiative of Excellence (Grant ``Investissements d'Avenir" ANR-16-IDEX-0008).

\bibliographystyle{abbrv}
\bibliography{Final_1_1_2022}
\end{document}